\documentclass[copyright]{eptcs}
\usepackage{times}
\usepackage{proof}
\usepackage{amssymb}
\usepackage{comment}
\usepackage{color}
\usepackage{amsmath}
\usepackage{listings}

\newtheorem{theorem}{Theorem}
\newtheorem{lemma}{Lemma}
\newtheorem{corollary}{Corollary}

\newcommand{\To}{\Rightarrow}
\newcommand{\next}[0]{\textit{next}}

\newcommand{\SN}[0]{\textit{SN}}
\newcommand{\vc}[0]{\texttt{vec}}
\newcommand{\nat}[0]{\texttt{nat}}
\newcommand{\Tvec}[0]{\texttt{T}^\vc}
\newcommand{\nil}[0]{\texttt{nil}}
\newcommand{\cons}[0]{\texttt{cons}}

\newcommand{\join}[0]{\texttt{join}}
\newcommand{\cast}[0]{\texttt{cast}}

\newcommand{\leadstov}[0]{\leadsto_v}

\newcommand{\ifzero}[0]{\texttt{ifZero}}
\newcommand{\impapp}[0]{\raisebox{-0.2ex}{$\square$}}

\newcommand{\interp}[1]{[\negthinspace[#1]\negthinspace]}
\newcommand{\drop}[1]{| #1 |}

\begin{document}

\pagestyle{empty}

\newif \ifTR \TRfalse

\title{Equality, Quasi-Implicit Products, and Large Eliminations}
\author{Vilhelm Sj\"oberg
\institute{Computer and Information Science\\
University of Pennsylvania}
\email{vilhelm@cis.upenn.edu}
\and
Aaron Stump 
\institute{Computer Science\\
The University of Iowa}
\email{astump@acm.org}}

\date{}

\def\authorrunning{Sj\"oberg and Stump}
\def\titlerunning{Equality and Implicit Products}

\maketitle
\thispagestyle{empty}

\begin{abstract}
This paper presents a type theory with a form of equality reflection:
provable equalities can be used to coerce the type of a term.
Coercions and other annotations, including implicit arguments, are
dropped during reduction of terms.  We develop the metatheory for an
undecidable version of the system with unannotated terms.  We then
devise a decidable system with annotated terms, justified in terms of
the unannotated system.  Finally, we show how the approach can be
extended to account for large eliminations, using what we call
quasi-implicit products.
\end{abstract}

\section{Introduction}
\label{sec:intro}

The main goal of this paper, as of several recent works, is to
facilitate external reasoning about dependently typed
programs~\cite{mishra-linger+08,barras+08}.  This is hampered if one
must reason about specificational data occurring in terms.
Specificational data are data which have no effect on the result of
the computation, and are present in program text solely for
verification purposes.  In traditional formal methods, specification
data are also sometimes called ghost data.  For example, consider the
familiar example of vectors $\langle\vc\ \phi\ l\rangle$ indexed by
both the type $\phi$ of the elements and the length $l$ of the vector.
An example dependently typed program is the $\textit{append}_\phi$
function (we work here with monomorphic functions, but will elide type
subscripts), operating on vectors holding data of type $\phi$.  We can
define \textit{append} so that it has the following type, assuming a
standard definition of $\textit{plus}$ on unary natural numbers
$\nat$:
\[
\textit{append} \ : \ \Pi l_1:\nat.\,\Pi l_2:\nat.\,\Pi v_1:\langle\vc\ \phi\ l_1\rangle.\,\Pi v_2:\langle\vc\ \phi\ l_2\rangle.\  \langle\vc\ \phi\ (\textit{plus}\ l_1\ l_2)\rangle
\]
\noindent We might wish to prove that \textit{append} is associative.
In type theories such as \textsc{Coq}'s Calculus of Inductive
Constructions, we would do this by showing that the following type is
inhabited:
\[
\begin{array}{l}
\Pi l_1:\nat.\, \Pi l_2:\nat.\,\Pi l_3:\nat.\,
\Pi v_1:\langle \vc\ \phi\ l_1\rangle.\, \Pi v_2:\langle \vc\ \phi\ l_2\rangle. \, \Pi v_3:\langle \vc\ \phi\ l_3\rangle.\\
\ \ \ \ (\textit{append}\ (\textit{plus}\ l_1\ l_2)\ l_3\ (\textit{append}\ l_1\ l_2\ v_1\ v_2)\ v_3) 
= (\textit{append}\ l_1\ (\textit{plus}\ l_2\ l_3)\ v_1\ (\textit{append}\ l_2\ l_3\ v_2\ v_3))
\end{array}
\]
\noindent Notice how the lengths of the vectors are cluttering even
the statement of this theorem.  Tools like \textsc{Coq} allow such
arguments to be elided, when they can be uniquely reconstructed.  So
the theorem to prove can be written in the much more palatable form:
\[
\begin{array}{l}
\Pi l_1:\nat.\, \Pi l_2:\nat. \, \Pi l_3:\nat.\,
\Pi v_1:\langle \vc\ \phi\ l_1\rangle.\, \Pi v_2:\langle \vc\ \phi\ l_2\rangle.\, \Pi v_3:\langle \vc\ \phi\ l_3\rangle.\\
\ \ \ \ (\textit{append}\ (\textit{append}\ v_1\ v_2)\ v_3) 
= (\textit{append}\ v_1\ (\textit{append}\ v_2\ v_3))
\end{array}
\]
\noindent This is much more readable.  But as others have noted, while
the indices have been elided, they are not truly erased.  This means
that the proof of associativity of \textit{append} must make use of
associativity also of \textit{plus}, in order for the lengths of the
two vectors (on the two sides of the equation) to be equal.  Indeed,
even stating this equation may require some care, since the types of
the two sides are not definitionally equal: one has
$(\textit{plus}\ (\textit{plus}\ l_1\ l_2)\ l_3)$ where the other has
$(\textit{plus}\ l_1\ (\textit{plus}\ l_2\ l_3))$.  This is where
techniques like heterogeneous equality come into
play~\cite{mcbride99}.

One solution to this problem is via intersection types, also called in
this setting \emph{implicit products}, as in the Implicit Calculus of
Constructions~\cite{miquel01}.  An implicit product $\forall
x:\phi.\phi'$ is the type for functions whose arguments are erased
during conversion (cf.~\cite{mishra-linger+08,barras+08}).  Such a
type can also be viewed as an infinite intersection type, since its
typing rule will assert $\Gamma \vdash t : \forall x:\phi.\phi'$
whenever $\Gamma, x :\phi\vdash t:\phi'$.  This rule formalizes
(approximately) the idea that $t$ is in the type $\forall
x:\phi.\phi'$ whenever it is in each instance of that type (i.e., each
type $[u/x]\phi'$ for $u:\phi$).  Thus, membership in the
$\forall$-type follows from membership in the instances of the body of
the $\forall$-type, making the $\forall$-type an intersection of those
instances.  Note that this is an infinitary intersection, and thus
different from the classical finitary intersection type
of~\cite{CoppoDezani78}.  We note in passing that the current work includes
first-class datatypes, while the other works just cited all rely on
encodings of inductive data as lambda terms.

We seek to take the previous approaches further, and erase not just
arguments to functions typed with implicit products, but all
annotations.  This is not the case in the Implicit Calculus of
Constructions, for example, or its algorithmic development
$\textit{ICC}^*$~\cite{barras+08}, where typing annotations other than
implicit arguments are not erased from terms.  When testing
$\beta$-equivalence of terms, we will work with unannotated versions
of those terms, where all type- and proof-annotations have been
dropped.  For associativity of \textit{append}, the proof does not
require associativity of \textit{plus}.  From the point of view of
external reasoning, \textit{append} on vectors will be
indistinguishable from \textit{append} on lists (without statically
tracked length).  

\textbf{The $\Tvec$ Type Theory.} This paper studies versions of a
type theory we call $\Tvec$.  This system is like G\"odel's System T,
with vectors and explicit equality proofs.  We first study an
undecidable version of $\Tvec$ with equality reflection, where terms
are completely unannotated (Section~\ref{sec:core}).  We establish
standard meta-theoretic results for this unannotated system
(Section~\ref{sec:metatheory}).  We then devise a decidable annotated
version of the language, which we also call $\Tvec$ (the context will
determine whether the annotated or unannotated language is intended).
The soundness of annotated $\Tvec$ is justified by erasure to the
unannotated system (Section~\ref{sec:tvec}).  We consider the
associativity of \textit{append} in annotated $\Tvec$, as an example
(Section~\ref{sec:eg}).  This approach of studying unannotated versus
annotated versions of the type theory should be contrasted with the
approach taken in NuPRL, based on Martin-L\"of's extensional type
theory~\cite{C86,ml84}.  There, one constructs typing derivations, as
separate artifacts, for unannotated terms.  Here, we unite the typing
derivation and the unannotated term in a single artifact, namely the
annotated term.

\textbf{Large eliminations.} Type-level computation poses challenges
for our approach.  Because coercions by equality proofs are erased
from terms, if we naively extended the system with large eliminations
(types defined by pattern matching on terms) we would be able
to assign types to diverging or stuck terms. We propose a solution
based on what we call \emph{quasi-implicit products}.  These
effectively serve to mark the introduction and elimination of the
intersection type, and prohibit call-by-value reduction within an
introduction.  This saves Normalization and Progress, which would
otherwise fail.  We develop the meta-theory of an extension of the
unannotated system with large eliminations and call-by-value
reduction, including normalization (Section~\ref{sec:tveclarge}).

The basic idea of basing provable equality on the operational
semantics of unannotated terms has been implemented previously in the
\textsc{Guru} dependently programming language, publicly available at
\url{http://www.guru-lang.org}~\cite{guru09}.  The current paper
improves upon the work on \textsc{Guru}, by developing and analyzing a
formal theory embodying that idea (lacking in~\cite{guru09}).

\section{Unannotated $\Tvec$}
\label{sec:core}

The definition of unannotated $\Tvec$ uses  unannotated terms $a$ (we
sometimes also write $b$):
\[
\begin{array}{lll}
 a & ::= & x\ |\ (a\ a')\ |\ \lambda x.a\ |\ 0\ |\ (S\ a)\ |\ (R_\nat\ a\ a'\ a'') \ 
|\ \nil\ |\ (\cons\ a\ a')\ |\ (R_\vc\ a\ a'\ a'')\ |\ \join\ 
\end{array}
\]
\noindent Here, $x$ is for $\lambda$-bound variables and $S$ is for
successor (not the $S$ combinator).  $R_\nat$ is the recursor over
natural numbers, and $R_\vc$ is the recursor over vectors.  We have
constructors $\nil$ and $\cons$ for vectors.  The term construct
$\join$ is the introduction form for equality proofs.  We will not
need an elimination form, since our system includes a form of equality
reflection.  For readability, we sometimes use meta-variable $l$ for
terms $a$ intended as lengths of vectors.  Types $\phi$ are defined
by:
\[
 \phi\ ::=\ \nat\ |\ \langle \vc\ \phi\ a\rangle\ |\ \Pi x:\phi.\phi'\ |\ \forall x:\phi.\phi'\ |\ a = a' 
\]
\noindent The first $\Pi$-type is as usual, while the second is an
intersection type abstracting a specificational $x$.  This $x$ need
not be $\lambda$-abstracted in the corresponding term, nor supplied as
an argument when that term is applied, similarly to Miquel's implicit
products~\cite{miquel01}. 

The reduction relation is the compatible closure under arbitrary
contexts of the rules in Figure~\ref{fig:red}.
Figure~\ref{fig:typing} gives type assignment rules for $\Tvec$, using
a standard definition of typing contexts $\Gamma$.  We define
$\Gamma\ \textit{Ok}$ to mean that if
$\Gamma\equiv\Gamma_1,x:\phi,\Gamma_2$, then
$\textit{FV}(\phi)\subset\textit{dom}(\Gamma_1)$. We use $a\downarrow
a'$ to mean that $a$ and $a'$ are joinable with respect to our
reduction relation (i.e., there exists $\hat{a}$ such that
$a\leadsto^*\hat{a}$ and $a'\leadsto^*\hat{a}$).

Perhaps surprisingly we do not track well-formedness of types, and
indeed the \texttt{join} and \texttt{conv} rules can introduce
untypable terms into types. However, they preserve the invariant
that terms deemed equal are joinable, and that turns out to
be enough to ensure type safety.

Type assignment is not syntax-directed, due to the \texttt{(conv)},
\texttt{(spec-abs)}, and \texttt{(spec-app)} rules, and not obviously
decidable.  This will not pose a problem here as we study the
meta-theoretic properties of the system.  Section~\ref{sec:tvec}
defines a system of annotated terms which is obviously decidable, and
justifies it by translation to unannotated $\Tvec$.  We work up to
syntactic identity modulo safe renaming of bound variables, which we
denote $\equiv$.

\begin{figure}
\[
\begin{array}{lll}
(\lambda x.a)\ a' & \leadsto & [a'/x]a \\
(R_\nat\ a\ a'\ 0) & \leadsto & a \\
(R_\nat\ a\ a'\ (S\ a'')) & \leadsto & (a'\ a''\ (R_\nat\ a\ a'\ a'')) \\
(R_\vc\ a\ a'\ \nil) & \leadsto & a \\
(R_\vc\ a\ a'\ (\cons\ a_1\ a'')) & \leadsto & (a'\ a_1\ a''\ (R_\vc\ a\ a'\ a'')) \\
\end{array}
\]
\caption{Reduction semantics for unannotated $\Tvec$ terms}
\label{fig:red}
\end{figure}

\begin{figure}
\[
\begin{array}{ll}
\infer[\texttt{var}]{\Gamma\vdash x:\phi}{\Gamma(x) \equiv \phi & \Gamma\,\textit{Ok}}
&
\ 
\\ \\

\infer[\texttt{join}]{\Gamma\vdash \join : a = a'}{a \downarrow a' & \Gamma\,\textit{Ok}}
&
\infer[\texttt{conv}]{\Gamma\vdash a:[a''/x]\phi}{\Gamma \vdash a''': a' = a'' & \Gamma \vdash a:[a'/x]\phi & x\not\in\textit{dom}(\Gamma)} 
\\ \\

\infer[\texttt{spec-abs}]{\Gamma \vdash a : \forall x:\phi'.\phi}{\Gamma, x:\phi' \vdash a:\phi & x\not\in\textit{FV}(a)}
&
\infer[\texttt{spec-app}]{\Gamma \vdash a : [a'/x]\phi}{\Gamma \vdash a:\forall x:\phi'.\phi & \Gamma \vdash a':\phi'}
\\ \\

\infer[\texttt{abs}]{\Gamma \vdash \lambda x.a : \Pi x:\phi'.\phi}
                    {\Gamma, x:\phi' \vdash a:\phi}
&
\infer[\texttt{app}]{\Gamma \vdash (a\ a') : [a'/x]\phi}{\Gamma \vdash a:\Pi x:\phi'.\phi & \Gamma\vdash a':\phi'}

\\ \\
\infer[\texttt{zero}]{\Gamma\vdash 0:\nat}{\Gamma\,\textit{Ok}} 
&
\infer[\texttt{nil}]{\Gamma\vdash \nil:\langle\vc\ \phi\ 0\rangle}{\Gamma\,\textit{Ok}} 

\\ \\

\infer[\texttt{succ}]{\Gamma\vdash (S\ a):\nat}{\Gamma\vdash a:\nat  }
&
\infer[\texttt{Rnat}]{\Gamma\vdash (R_\nat\ a\ a'\ a''):[a''/x]\phi}
      {\begin{array}{l} 
       x\not\in\textit{dom}(\Gamma) \\
       \Gamma \vdash a'' : \nat \\
       \Gamma \vdash a : [0/x]\phi \\
       \Gamma \vdash a' : \Pi y:\nat. \Pi u : [y/x]\phi. [(S y)/x]\phi
       \end{array}}
\\ \\

\infer[\texttt{cons}]{\Gamma\vdash (\cons\ a\ a'):\langle \vc\ \phi\ (S\ l)\rangle}
      {\begin{array}{l}\Gamma\vdash a:\phi \\ \Gamma \vdash a':\langle \vc\ \phi\ l\rangle
       \end{array}}
&
\infer[\texttt{Rvec}]{\Gamma\vdash (R_\vc\ a\ a'\ a''):[l/y, a''/x]\phi}
      {\begin{array}{l}
       x\not\in\textit{dom}(\Gamma) \\
       \Gamma \vdash a'' : \langle \vc\ \phi'\ l\rangle \\
       \Gamma \vdash a : [0/y,\nil/x]\phi \\
       \Gamma \vdash a' : \Pi z:\phi'. \forall l:\nat. \Pi v :\langle \vc\ \phi'\ l\rangle. \Pi u : [l/y, v/x]\phi. \\
        \ \ \ \ \ \ \ \ \ \ \ \ \  [(S\ l)/y, (\cons\ z\ v)/x]\phi
       \end{array}}

\end{array}
\]
\caption{Type assignment system for unannotated $\Tvec$}
\label{fig:typing}
\end{figure}

\section{Metatheory of Unannotated $\Tvec$}
\label{sec:metatheory}

$\Tvec$ enjoys standard properties: Type Preservation, Progress (for
closed terms), and Strong Normalization.  These are all easily
obtained, the last by dependency-erasing translation to another type
theory (as done originally for LF in~\cite{HHP93}).  Here, we consider
a more semantically informative approach to Strong Normalization.
Omitted proofs may be found in a companion report on the second
author's web page (see \url{http://www.cs.uiowa.edu/~astump/papers/ITRS10-long.pdf}).

\begin{theorem}[Type Preservation]
\label{thm:tppres}
If $\Gamma \vdash a : \phi$ and $a\leadsto a'$, then $\Gamma \vdash a':\phi$.
\end{theorem}

\begin{theorem}[Progress]
\label{thm:progress}
If $\Gamma \vdash a : \phi$ and $dom(\Gamma) \cap FV(a) = \emptyset$, then either $a$ is a value or  $\exists a'. a\leadsto a'$. Here a \emph{value} is a term of the form
\[
\begin{array}{lll}
v & ::= & \lambda x.a\ |\ 0\ |\ (S\ v)\ |\ \nil\ |\ (\cons\ v\ v')\ |\ \join
\end{array}
\]
\end{theorem}

\subsection{Semantics of equality}

For our Strong Normalization proof, a central issue is providing an
interpretation for equality types in the presence of free variables.
We would like to interpret equations like $(\textit{plus}\ 2\ 2) = 4$
(where the numerals abbreviate terms formed with $S$ and $0$ as usual,
and \textit{plus} has a standard recursive definition), as simply
$(\textit{plus}\ 2\ 2) \downarrow 4$.  But when the two terms contain
free variables -- e.g., in $(\textit{plus}\ x\ y) =
(\textit{plus}\ y\ x)$ -- or when the context is inconsistent, the
semantics should make the equation true, even though its sides are not
joinable.  So our semantics for equality types is joinability under
all \emph{ground instances} of the context $\Gamma$.  The notation for
this is $a \sim_\Gamma a'$.  The definition must be given as part of
the definition of the interpretation of types, because we want to
stipulate that the substitutions $\sigma$ replace each variable $x$ by
a ground term in the interpretation of $\sigma \Gamma(x)$.  When
$\Gamma$ is empty, we will write $a\sim_\Gamma a'$ as $a\sim a'$.  We
use a similar convention for other notations subscripted by a context
below.

\subsection{The interpretation of types}

The interpretation of types is given in Figure~\ref{fig:interp}.  In
that figure, we write $\To$ and $\Leftrightarrow$ for meta-level
implication and equivalence, respectively, and give $\Leftrightarrow$
lowest precedence among all infix symbols, and $\To$ next lowest
precedence.  We stipulate up front (not in the clauses in the figure)
that $a\in\interp{\phi}_\Gamma$ requires $a\in\SN$ (where $\SN$ is the
set of strongly normalizing terms) and $\Gamma\vdash a:\phi$.  The
definition in Figure~\ref{fig:interp} proceeds by well-founded
recursion on the triple $(|\Gamma|,d(\phi),l(a))$, in the natural
lexicographic ordering.  Here, $|\Gamma|$ is the cardinality of
$\textit{dom}(\Gamma)$, and if $a\in\SN$, then we make use of a
(finite) natural number $l(a)$ bounding the number of symbols in the
normal form of $a$.  We need to assume confluence of reduction
elsewhere in this proof, so it does not weaken the result to assume
here that each term has at most one normal form.  While we believe
confluence for this language should be easily established by standard
methods, that proof remains to future work.  The quantity $d(\phi)$ is
the depth of $\phi$, defined as follows:
\[
\begin{array}{lllllll}
d(\nat) & = & 0 &\ \ \ \ \ &
d(\langle\vc\ \phi\ l\rangle) & = & 1+d(\phi) \\
d(\Pi x:\phi.\phi') & = & 1+\textit{max}(d(\phi),d(\phi')) & \ \ \ \ \ &
d(\forall x:\phi.\phi') & = & 1+\textit{max}(d(\phi),d(\phi'))\\
d(a = a') & = & 0 & \ &\ &\ &\ 
\end{array}
\]
\noindent Note that $d(\phi)=d([a/x]\phi)$ for all $a$, $x$, and
$\phi$.  Also, in the clause for $\vc$-types, since the right hand
side of the clause conjoins the condition $a\in\SN$, $l(a)$ is
defined, and we have $l(a'')<l(\cons\ a'\ a'')$.  The figure gives an
inductive definition for when $\sigma\in\interp{\Gamma}_\Delta$.  We
call such a $\sigma$ a \emph{closable substitution}.

\begin{figure}
\[
\begin{array}{l}
\begin{array}{lll}
a \in \interp{\nat}_\Gamma & \Leftrightarrow & \top\\
a \in \interp{\langle\vc\ \phi\ l\rangle}_\Gamma & \Leftrightarrow & (a\leadsto^* \nil\ \To\  l \sim_\Gamma 0) \ \wedge\\
\ &\ & \begin{array}{ll}
       \forall a'.\, \forall a''.\, a\leadsto^* (\cons\ a'\ a'')\ \To  &
      (i)\ a'\in\interp{\phi}_\Gamma\ \wedge\ \exists l'. \\
\ &   (ii)\ a''\in\interp{\langle\vc\ \phi\ l'\rangle}_\Gamma\ \wedge\\ 
\ &   (iii)\ l \sim_\Gamma (S\ l')
\end{array}\\
a\in\interp{\Pi x:\phi'.\phi}_\Gamma & \Leftrightarrow & \forall a'\in\interp{\phi'}^+_\Gamma.\ (a\ a')\in\interp{[a'/x]\phi}_\Gamma \\
a\in\interp{\forall x:\phi'.\phi}_\Gamma & \Leftrightarrow & \forall a'\in\interp{\phi'}^+_\Gamma.\ a\in\interp{[a'/x]\phi}_\Gamma \\
a \in \interp{a_1 = a_2}_\Gamma & \Leftrightarrow & (a \leadsto^* \join \To a_1 \sim_\Gamma a_2) \\ 
\end{array}\\ \\
\textnormal{\underline{where}:}\\
\begin{array}{lll}

a \sim_\Gamma a' & \Leftrightarrow & \forall \sigma.\ \sigma\in\interp{\Gamma}\ \Rightarrow\ (\sigma a) \downarrow (\sigma a') \\ 

a\in\interp{\phi}^+_{\Gamma} & \Leftrightarrow & a\in\interp{\phi}_\Gamma
\ \wedge \ (|\Gamma| > 0 \ \To\ \forall \sigma\in\interp{\Gamma}.\ \sigma a\in\interp{\sigma \phi})
\end{array} \\ \\ 
\textnormal{\underline{and also}:} \\ 
\begin{array}{ll}
\infer{\emptyset\in\interp{\cdot}_{\Delta}}{\ }
&
\infer{\sigma\cup\{(x,a)\}\in\interp{\Gamma,x:\phi}_{\Delta}}
      {a\in\interp{\sigma\phi}^+_{\Delta} & \sigma \in\interp{\Gamma}_{\Delta}}
\end{array}
\end{array}
\]
\caption{The interpretation $a\in\interp{\phi}_\Gamma$ of strongly normalizing terms with $\Gamma\vdash a:\phi$}
\label{fig:interp}
\end{figure}

In general, the inductive definition of closable substitution
$\sigma\in\interp{\Gamma}_{\Delta}$ allows the range of the
substitution to contain open terms.  When $\Delta$ is empty, $\sigma$
is a \emph{closing} substitution.  The definition of $\interp{\cdot}$
for types uses the definition of closable substitutions in a
well-founded way.  We appeal only to $\interp{\Gamma}$ (with an empty
context $\Delta$) in the definitions of $\interp{\phi}_\Gamma$ and
$\interp{\phi}^+_\Gamma$.  Where the definition of
$\interp{\Gamma}_\Delta$ appeals back to the interpretation of types,
it does so only when this $\Gamma$ was non-empty, and with an empty
context given for the interpretation of the type.  So $|\Gamma|$ has
indeed decreased from one appeal to the interpretation of types to the
next.

\subsection{Critical properties}
\label{sec:critprop}

A term is defined to be \emph{neutral} iff it is of the form $(a\ a')$
or $(R_B\ a\ a'\ a'')$ (with $B \in \{\nat, \vc\}$), or if it is a
variable.  We prove three critical properties of reducibility at type
$\phi$, by mutual induction on $(|\Gamma|,d(\phi),l(a))$.  Here we
write $\next(a) = \{ a' \mid a \leadsto a' \}$.

\ 

\noindent \textbf{R-Pres}. If $a\in\interp{\phi}_\Gamma$, then $\next(a)\subset\interp{\phi}_\Gamma$.

\noindent \textbf{R-Prog}. If $a$ is neutral and $\Gamma\vdash
a:\phi$, then
$\next(a)\subset\interp{\phi}_\Gamma$ implies $a\in\interp{\phi}_\Gamma$.

\noindent \textbf{R-Join}. Suppose $a_1 \sim_\Gamma a_2$;
$\Gamma\vdash a':a_1 = a_2$ for some $a'$; and
$x\not\in\textit{dom}(\Gamma)$. Then $\interp{[a_1/x]\phi}_\Gamma\subset
\interp{[a_2/x]\phi}_\Gamma$.

\subsection{Soundness of typing with respect to the interpretation}

Our typing rules are sound with respect to our interpretation of types
(Figure~\ref{fig:interp}).  As usual, we must strengthen the statement
of soundness for the induction to go through.  We need a quasi-order
$\subset$ on contexts, defined by:
$\Delta\subset\Gamma\ \Leftrightarrow\ \forall
x\in\textit{dom}(\Delta).\ \Delta(x) = \Gamma(x)$.

\begin{theorem}[Soundness for Interpretations]
\label{thm:soundness}
Suppose $\Gamma\vdash a:\phi$.  Then for any $\Delta\,\textit{Ok}$
with $\Delta\subset\Gamma$ and $\sigma\in\interp{\Gamma}_{\Delta}$, we
have $(\sigma a)\in\interp{\sigma \phi}_{\Delta}$.  
\end{theorem}

\noindent Critically, we quantify over possibly open substitutions
$\sigma$, whose ranges consist of closable terms.  

\begin{corollary}[Strong Normalization]
\label{cor:sn}
If $\Gamma \vdash a : \phi$, then $a\in\SN$.
\end{corollary}

\begin{corollary}
\label{cor:eqdec}
If $\Gamma \vdash a : \phi$ and $\Gamma \vdash a' : \phi'$, then $a \downarrow a'$ is decidable.
\end{corollary}

\begin{corollary}[Equational Soundness]
\label{cor:eqsnd}
If $\cdot \vdash a:b_1 = b_2$, then $b_1 \downarrow b_2$.
\end{corollary}

\begin{corollary}[Logical Soundness]
There is a type $\phi$ such that $\vdash a:\phi$ does not hold for any $a$.
\end{corollary}

\noindent \textbf{Proof.} By Equational Soundness, we do not have $\vdash a:0 = (S\ 0)$ for any $a$.

\section{Annotated $\Tvec$}
\label{sec:tvec}

We now define a system of annotated terms $t$, and a decidable type
computation system deriving judgments $\Gamma \Vdash t:\phi$,
justified by dropping annotations via $\drop{\cdot}$ (defined in
Figure~\ref{fig:drop}).  The annotated terms $t$ are the following.
Annotations include types $\phi$, possibly with designated free
variables, as in $x.\phi$ (bound by the dot notation).
\[
\begin{array}{lll}
 t & ::= & x\ |\ (t\ t')\ |\ (t\ t')^-\ |\ \lambda x:\phi.t\ |\ \lambda^- x:\phi.t\ |\ 0\ |\ (S\ t)\ |\ (R_\nat\ x.\phi\ t\ t'\ t'') \\
\ & \ & |\ (\nil\ \phi)\ |\ (\cons\ t\ t')\ |\ (R_\vc\ x.y.\phi\ t\ t'\ t'')\ |\ (\join\ t\ t')\ |\ (\cast\ x.\phi\ t\ t') 
\end{array}
\]
\noindent Three new constructs correspond to the typing rules
$\texttt{(spec-abs)}$, $\texttt{(spec-app)}$, and $\texttt{(conv)}$ of
Figure~\ref{fig:typing}: $\lambda^- x:\phi'.\phi$, $(t\ t')^-$ and
$(\cast\ x.\phi\ t\ t')$.  Figure~\ref{fig:tpcomp} gives
syntax-directed type-computation rules, which constitute a
deterministic algorithm for computing a type $\phi$ as output from a
context $\Gamma$ and annotated term $t$ as inputs.  Several rules use
the $\drop{\cdot}$ function, since types $\phi$ (as defined in
Section~\ref{sec:core} above) may mention only unannotated terms.

\begin{theorem}[Algorithmic Typing]
Given $\Gamma$ and $a$, we can, in an effective way, either find
$\phi$ such that $\Gamma \Vdash a : \phi$, or else report that there
is no such $\phi$.
\end{theorem}

\noindent This follows in a standard way from inspection of the rules,
using Corollary~\ref{cor:eqdec} for the $\join$-rule.

\begin{theorem}[Soundness for Type Assignment]
If $\Gamma \Vdash t:\phi$ then $\Gamma \vdash \drop{t} : \phi$.
\end{theorem}

\begin{figure}
\[
\begin{array}{lllllll}
\drop{x} & = & x &\ \ \ \ \ & 
\drop{(t\ t')} & = & (\drop{t}\ \drop{t'}) \\
\drop{(t\ t')^-} & = & \drop{t} &\ \ \ &
\drop{\lambda x:\phi.t} & = & \lambda x.\drop{t} \\
\drop{\lambda^- x:\phi.t} & = & \drop{t} &\ \ \ &
\drop{0} & = & 0 \\
\drop{(S\ t)} & = & (S\ \drop{t}) &\ \ \ &
\drop{(\nil\ \phi)} & = & \nil \\
\drop{(\cons\ t\ t')} & = & (\cons\ \drop{t}\ \drop{t'}) &\ \ \ &
\drop{(R_\nat\ x.\phi\ t\ t'\ t'')} & = & (R_\nat\ \drop{t}\ \drop{t'}\ \drop{t''})\\
\drop{(R_\vc\ x.y.\phi\ t\ t'\ t'')} & = & (R_\vc\ \drop{t}\ \drop{t'}\ \drop{t''}) &\ \ \ &
\drop{(\join\ t\ t')} & = & \join \\
\drop{(\cast\ x.\phi\ t\ t')} & = & \drop{t'} &\ &\ &\ &\  
\end{array}
\]
\caption{Translation from annotated terms to unannotated terms}
\label{fig:drop}
\end{figure}

\begin{figure}
\[
\begin{array}{lll}
\infer{\Gamma\Vdash (\join\ t\ t') : \drop{t} = \drop{t'}}{\Gamma \Vdash t:\phi &  \Gamma \Vdash t':\phi' & \drop{t} \downarrow \drop{t'}}
&
\infer{\Gamma\Vdash (\cast\ x.\phi\ t\ t'):[a'/x]\phi}{\Gamma \Vdash t: a = a' & \Gamma \Vdash t':[a/x]\phi} 
&
\infer{\Gamma \Vdash \lambda^- x:\phi'.t : \forall x:\phi'.\phi}{\Gamma, x:\phi' \Vdash t:\phi & x\not\in\textit{FV}(\drop{t})}
\\ \\
\infer{\Gamma \Vdash (t\ t')^- : [\drop{t'}/x]\phi}{\Gamma \Vdash t:\forall x:\phi'.\phi & \Gamma \Vdash t':\phi'}
&

\infer{\Gamma \Vdash \lambda x:\phi'. t : \Pi x:\phi'.\phi}
      {\Gamma, x:\phi' \Vdash t:\phi}
&
\infer{\Gamma \Vdash (t\ t') : [\drop{t'}/x]\phi}{\Gamma \Vdash t:\Pi x:\phi'.\phi & \Gamma\Vdash t':\phi'}

\\ \\
\multicolumn{3}{c}{\infer{\Gamma\Vdash (R_\vc\ x.y.\phi\ t\ t'\ t''):[l/x, \drop{t''}/y]\phi}
      {\begin{array}{l}\Gamma \Vdash t'' : \langle \vc\ \phi'\ l\rangle \\
       \Gamma \Vdash t : [0/x,\nil/y]\phi \\
       \Gamma \Vdash t' : \forall l:\nat. \Pi z:\phi'. \Pi v :\langle \vc\ \phi'\ l\rangle. \Pi u : [l/x, v/y]\phi. \\
        \ \ \ \ \ \ \ \ \ \ \ \ \  [(S\ l)/x, (\cons\ z\ v)/y]\phi
       \end{array}}}
\end{array}
\]
\caption{Type-computation system for annotated $\Tvec$ (selected rules)}
\label{fig:tpcomp}
\end{figure}

\subsection{Example}
\label{sec:eg}

Now let us see versions of the examples mentioned in
Section~\ref{sec:intro}, available in the \texttt{guru-lang/lib/vec.g}
library file for \textsc{Guru} (see \url{www.guru-lang.org}).  The
desired types for vector append (``\textit{append}'') and for
associativity of vector append are:
\[
\begin{array}{lll}
\textit{append}& : &\forall l_1:\nat. \forall l_2:\nat. \Pi v_1:\langle \vc\ \phi\ l_1\rangle. \Pi v_2:\langle \vc\ \phi\ l_2\rangle. \langle\vc\ \phi\ (\textit{plus}\ l_1\ l_2)\rangle \\ 
\textit{append\_assoc}& :& \forall l_1:\nat. \forall l_2:\nat.\forall l_3:\nat. \\
\ & \ & \Pi v_1:\langle \vc\ \phi\ l_1\rangle. \Pi v_2:\langle \vc\ \phi\ l_2\rangle. \Pi v_3:\langle \vc\ \phi\ l_3\rangle.\\
\ & \ & \ \ \ (\textit{append}\ (\textit{append}\ v_1\ v_2)\ v_3) = (\textit{append}\ v_1\ (\textit{append}\ v_2\ v_3))
\end{array}
\]
\noindent We consider now annotated inhabitants of these types.  The
first is the following:
{\small
\[
\begin{array}{lll}
\textit{append} &=& \lambda^- l_1:\nat. \lambda^-l_2:\nat. \lambda v_1:\langle \vc\ \phi\ l_1\rangle. \lambda v_2:\langle \vc\ \phi\ l_2\rangle. \\
\ &\ & \ \ \ \ (R_\vc\ (x.y.\langle\vc\ \phi\ (\textit{plus}\ x\ l_2)\rangle)\ \\
\ &\ & \ \ \ \ \ \ \ \ \ \ (\cast\ (x.\langle\vc\ \phi\ x\rangle)\ P_1\ v_2) \\
\ &\ & \ \ \ \ \ \ \ \ \ \ (\lambda^- l:\nat. \lambda x:\phi.\lambda v_1':\langle\vc\ \phi\ l\rangle.
                      \lambda r:\langle\vc\ \phi\ (\textit{plus}\ l\ l_2)\rangle).\\
\ &\ &\ \ \ \ \ \ \ \ \ \ \ \ \ \ \ (\cast\ (x.\langle\vc\ \phi\ x\rangle)\ P_2\ (\cons\ x\ r))\\
\ &\ & \ \ \ \ \ \ \ \ \ \ v_1)
\end{array}
\]}
\noindent The two cases in the $R_\vc$ term return a type-cast version
of what would standardly be returned in an unannotated version of
\textit{append}.  The proofs $P_1$ and $P_2$ used in those casts show
respectively that $l_2 = (\textit{plus}\ 0\ l_2)$ and
$(S\ (\textit{plus}\ l\ l_2)) = (\textit{plus}\ (S\ l)\ l_2)$. They
are simple join-proofs:
\[
\begin{array}{lllllll}
P_1 & = & (\join\ l_2\ (\textit{plus}\ 0\ l_2))&\ \ \ \ \ &
P_2 & = & (\join\ (S\ (\textit{plus}\ l\ l_2))\ (\textit{plus}\ (S\ l)\ l_2))
\end{array}
\]
\noindent Now for \textit{append\_assoc}, we can use the following annotated term:
{\small
\[
\begin{array}{lll}
\textit{append\_assoc} &=& \lambda^- l_1:\nat. \lambda^- l_2:\nat.\lambda^- l_3:\nat. \\
\ &\ &\lambda v_1:\langle \vc\ \phi\ l_1\rangle. \lambda v_2:\langle \vc\ \phi\ l_2\rangle. \lambda v_3:\langle \vc\ \phi\ l_3\rangle.\\
\ &\ &\ \ \ \ (R_\vc\ (x.y.(\textit{append}\ (\textit{append}\ v_1\ v_2)\ v_3) = (\textit{append}\ v_1\ (\textit{append}\ v_2\ v_3)))\ \\
\ &\ &\ \ \ \ \ \ \ \ \ \ (\join\ (\textit{append}\ (\textit{append}\ \nil\ v_2)\ v_3) = (\textit{append}\ \nil\ (\textit{append}\ v_2\ v_3))) \\
\ &\ &\ \ \ \ \ \ \ \ \ \ (\lambda^- l:\nat. \lambda x:\phi.\lambda v_1':\langle\vc\ \phi\ l\rangle.\\
\ &\ &\ \ \ \ \ \ \ \ \ \ \ \ \lambda r:(\textit{append}\ (\textit{append}\ v_1'\ v_2)\ v_3) = (\textit{append}\ v_1'\ (\textit{append}\ v_2\ v_3)).\\
\ &\ &\ \ \ \ \ \ \ \ \ \ \ \ \ \ \ P_3))
\end{array}
\]}
\noindent The omitted proof $P_3$ is an easy equational proof of the following type:
\[
(\textit{append}\ (\textit{append}\ (\cons\ x\ v_1')\ v_2)\ v_3) = (\textit{append}\ (\cons\ x\ v_1')\ (\textit{append}\ v_2\ v_3))
\]

\section{$\Tvec$ with Large Eliminations}
\label{sec:tveclarge}

\begin{figure}
\[ 
\begin{array}{lll}
\phi\ ::=\ \dots\ |\ \ifzero\ a\ \phi\ \phi' & \qquad a\ ::=\ \dots\ |\ \lambda.a\ |\ a\ \impapp   & \qquad v\ ::=\ \dots\ |\ \lambda.a
\end{array}
\]
\[
\begin{array}{ll}

\infer[\texttt{spec-abs'}]{\Gamma \vdash \lambda.a : \forall x:\phi'.\phi}{\Gamma, x:\phi' \vdash a:\phi & x\not\in\textit{FV}(a)}
&
\infer[\texttt{spec-app'}]{\Gamma \vdash a\ \impapp : [a'/x]\phi}{\Gamma \vdash a:\forall x:\phi'.\phi & \Gamma \vdash a':\phi'}

\\ \\

\infer[\texttt{foldZ}]
{\Gamma\vdash a : \ifzero\ 0\ \phi\ \phi'}
{\Gamma\vdash a : \phi}

&
\infer[\texttt{unfoldZ}]
{\Gamma\vdash a : \phi}
{\Gamma\vdash a : \ifzero\ 0\ \phi\ \phi'}

\\ \\

\infer[\texttt{foldS}]
{\Gamma\vdash a : \ifzero\ (S\ a')\ \phi\ \phi'}
{\Gamma\vdash a : \phi' & \Gamma\vdash a':\nat}

&
\infer[\texttt{unfoldS}]
{\Gamma\vdash a : \phi'}
{\Gamma\vdash a : \ifzero\ (S\ a')\ \phi\ \phi' & \Gamma\vdash a':\nat}
\end{array}
\]
\caption{Types, terms, values, and typing rules for $\Tvec$ with large eliminations.}
\label{fig:tveclarge}
\end{figure}

Next we study an extended version of $\Tvec$ with large eliminations,
i.e. types defined by pattern matching on terms.  This extended
language no longer is normalizing under general $\beta$-reduction
$\leadsto$, but we will prove that well-typed closed terms normalize
under call-by-value evaluation $\leadstov$. In particular, the language is
type safe and logically consistent.

The additions to the language and type system are shown in
figure~\ref{fig:tveclarge}. 

The type language is extended with the simplest possible form of large elimination,
a type-level conditional \ifzero which is introduced and
eliminated by the \texttt{fold} and \texttt{unfold} rules.  While type
conversion and type folding/unfolding are completely implicit, we
replace the \texttt{spec-abs/app} rules with new rules
\texttt{spec-abs'/app'} which require the place where we introduce or
eliminate the $\forall$-type to be marked by new \emph{quasi-implicit}
forms $\lambda.a$ and $a\ \impapp$. These forms do not mention the
quantified variable or the term it is instantiated with, so we retain
the advantages of specificational reasoning. The point of these forms
is their evaluation behavior: $(\lambda.a)\ \impapp \leadstov a$, and
$\lambda.a$ counts as a value so CBV evaluation will never reduce
inside it. Besides this, the CBV operational semantics is standard, so
we omit it here.

In the language with large eliminations we no longer have
normalization or type safety for arbitrary open terms. This is because
the richer type system lets us make use of absurd equalities: whenever we
have $\Gamma \vdash a : \phi$ and $\Gamma \vdash p : (S\ a')\!=\!0$, we can
show $\Gamma \vdash a : \phi'$ for any $\phi'$ by going via the
intermediate type $(\ifzero\ 0\ \phi (\alpha.\phi'))$. In particular, this
means we  can show judgments like
\[ 
p:1\!\!=\!0 \vdash (\lambda  x.x\ x)\ (\lambda x.x\ x) : \nat
\qquad\text{and}\qquad 
p:1\!\!=\!0 \vdash 0\ 0 : \nat. 
\]
This is also the reason we introduce the quasi-implicit
products. Using our old rule \texttt{spec-abs} we would be able to
show $\vdash 0\ 0 : \forall p:1\!\!=\!0.\nat$, despite $0\ 0$ being a
stuck term in our operational semantics.

Because of this \emph{quod libet} property it is no longer convenient
to prove Progress and Preservation before Normalization. While the
proof of Preservation is not hard, Progress as we have seen depends on
the logical consistency of the language, which is exactly what we hope
to establish through Normalization. To cut this circle we design an
interpretation of types (figure \ref{fig:large_interp}) that lets us
prove type safety, Canonical Forms and Normalization in a single
induction.

\subsection{Semantics of Equality}
We need to pick an interpretation for equality types. Since we
are only interested in closed terms, this can be less elaborate than
in section~\ref{sec:metatheory}. Perhaps surprisingly, even though we
are interested in CBV-evaluation of programs, we can still interpret
equality as joinability $\downarrow$ under unrestricted $\beta$-reduction. 
In the interpretation we use $\leadstov$ for the
program being evaluated, but $\leadsto$ whenever we talk about terms
occurring in types (namely in $\vc$, $=$, and R-types). The
\texttt{join} typing rule is specified in terms of $\leadsto$, so
when doing symbolic evaluation of programs at type checking time the 
type checker can use unrestricted reduction, which gives a powerful
type system than can prove many equalities.

\subsection{Normalization to Canonical Form}

\begin{figure}
\[
\begin{array}{ll}
\begin{array}{lll}
a \in \interp{\nat} & \Leftrightarrow & \exists n. a \leadstov^* n \\
a \in \interp{\langle\vc\ \phi\ l\rangle} & \Leftrightarrow & (a\leadstov^* \nil\ \land  l \leadsto^* 0) \ \lor \\
 & & \begin{array}{ll} \exists v\ v'\ n. & a\leadstov^* (\cons\ v\ v') \land l \leadsto^* (S\ n) \\
\ & \land\ v \in \interp{\phi} 
\ \land\ v'\in \interp{\langle\vc\ \phi\ n\rangle}\\
\end{array}\\
a\in\interp{\Pi x:\phi'.\phi} & \Leftrightarrow &  \exists a'.a \leadstov^* (\lambda x.a')
\ \land\  \forall a'\in\interp{\phi'}.\ (a\ a')\in\interp{[a'/x]\phi} \\
a\in\interp{\forall x:\phi'.\phi} & \Leftrightarrow &  \exists a'.a \leadstov^* (\lambda a')
\ \land\ \forall a'\in\interp{\phi'}.\ (a\ \impapp)\in\interp{[a'/x]\phi} \\
a \in \interp{a_1 = a_2} & \Leftrightarrow &  a \leadstov^* \join \ \land\ a_1 \downarrow a_2 \\
a \in \interp{\ifzero\ b\ \phi\ \phi'} & \Leftrightarrow & \begin{cases} a \in \interp{\phi} & \text{if } b \leadsto^* 0  \\ 
                                                                         a \in \interp{\phi'} & \text{if } b \leadsto^* (S\ n) \\ 
                                                                         \text{False} & \text{otherwise} \end{cases} \\
\end{array}
\begin{array}{c}
\infer{\emptyset\in\interp{\cdot}}{\ } \\
\\
\infer{\sigma\cup\{(x,v)\}\in\interp{\Gamma,x:\phi}}
      {v\in\interp{\sigma\phi} & \sigma \in\interp{\Gamma}}
\end{array}
\end{array}
\]
\caption{Type interpretation $a\in\interp{\phi}$ and context interpretation $\sigma \in \interp{\Gamma}$ for $\Tvec$ with large eliminations}
\label{fig:large_interp}
\end{figure}

We define the interpretation $\interp{\ }$ as in
figure~\ref{fig:large_interp} by recursion on the depth of the
type $\phi$. As we only deal with closed terms, the definition
can be simpler than the one in section \ref{sec:metatheory}.
The proof then proceeds much like the proof for open
terms:

\noindent \textbf{R-Canon}. If $a\in\interp{\phi}$, then $a\leadstov^*
v$ for some $v$. Furthermore, if the top-level constructor of $\phi$
is $\nat$, $\Pi$, $\forall$, $=$, or $\vc$, then $v$ is the
corresponding introduction form.

\noindent \textbf{R-Pres}. If $a\in\interp{\phi}$ and $a \leadstov a'$, then $a' \in \interp{\phi}$.

\noindent \textbf{R-Prog}. If $a \leadstov a'$, and $a' \in \interp{\phi}$, then $a\in\interp{\phi}$.

\noindent \textbf{R-Join}. If $a_1 \downarrow a_2$, then $a \in \interp{[a_1/x]\phi}$ 
implies $a \in \interp{[a_2/x]\phi}$.

\begin{theorem}
\label{thm:fundamental_tveclarge}
If $\Gamma \vdash a : \phi$ and $\sigma \in \interp{\Gamma}$,
then $\sigma a \in \interp{\sigma\phi}$.
\end{theorem}

\begin{corollary}[Type Safety]
\label{cor:sn}
If $\vdash a : \phi$, then $a \leadstov^* v$.
\end{corollary}

\begin{corollary}[Logical Soundness]
$\vdash a:1\!\!=\!0$ does not hold for any $a$.
\end{corollary}

\section{Conclusion and Future Work}

The $\Tvec$ type theory includes intersection types and a form of
equality reflection, justified by translation to an undecidable
unannotated system.  The division into annotated and unannotated
systems enables us to reason about terms without annotations, while
retaining decidable type checking.  We have seen how this approach
extends to a language including large eliminations, by introducing a
novel kind of \emph{quasi-implicit} products. The quasi-implicit
products allow convenient reasoning about specificational data, while
permitting a simple proof of normalization of closed terms.  Possible
future work includes formalizing the metatheory, and extending to a
polymorphic type theory.  Adding an extensional form of equality while
retaining decidability would also be of interest, as
in~\cite{altenkirch+07}.

\textbf{Acknowledgments:} Thanks to members of the \textsc{Trellys}
team, especially Stephanie Weirich and Tim Sheard, for discussions on
this and related systems.  This work was partially supported by the
the U.S. National Science Foundation under grants 0910510 and 0910786.

\bibliographystyle{eptcs}
\bibliography{biblio}

\ifTR
\appendix

\section{Proof of Type Preservation (Theorem~\ref{thm:tppres})}

\textbf{More about contexts.}  In more detail, we consider a contexts
$\Gamma$ to be a function from a finite set of variables to types,
together with a total ordering on its domain, and write
$\Gamma,x:\phi$ for the function that behaves just like $\Gamma$,
except that it returns $\phi$ for $x$, and places $x$ after the
variables in $\textit{dom}(\Gamma)$.  When
$\Gamma\equiv\Gamma_1\cup\Gamma_2$ with all the variables in
$\Gamma_2$ greater than those of $\Gamma_1$ in the ordering, we write
$\Gamma_1,\Gamma_2$ (implying also that the domains of $\Gamma_1$ and
$\Gamma_2$ are disjoint).

The proof of Type Preservation is by induction on the structure of the
assumed typing derivation.  We list all cases.  Unless we introduce
meta-variable $b$ for another purpose, in each case we will assume the
term in question reduces to $b$.  In cases where the term in question
is a normal form, this will lead to a contradiction.

\ 

\noindent \textbf{Case:}

\

$\infer{\Gamma\vdash x:\phi}{\Gamma(x) \equiv \phi}$

\ 

\noindent This case cannot arise, since $x$ is a normal form and so
cannot reduce.

\ 

\noindent \textbf{Case:}

\

$\infer{\Gamma\vdash \join : a = a'}{a \downarrow a'}$

\ 

\noindent This case cannot arise, since $\join$ is a normal form.

\ 

\noindent \textbf{Case:}

\

$\infer{\Gamma\vdash a:[a''/x]\phi}{\Gamma \vdash a''': a' = a'' & \Gamma \vdash a:[a'/x]\phi & x\not\in\textit{dom}(\Gamma)}$

\ 

\noindent (Recall that by convention in this proof, our second
assumption is $a\leadsto b$.)  By the induction hypothesis, we have
$\Gamma \vdash b:[a'/x]\phi$.  We may then reapply this rule to
conclude $\Gamma\vdash b:[a''/x]\phi$.

\ 

\noindent \textbf{Case:}

\

$\infer{\Gamma \vdash a : \forall x:\phi'.\phi}{\Gamma, x:\phi' \vdash a:\phi & x\not\in\textit{FV}(a)}$

\ 

\noindent By the induction hypothesis, we have $\Gamma, x:\phi' \vdash
b:\phi$.  Reduction cannot increase the set of free variables, so
$x\not\in\textit{FV}(b)$.  We may then reapply this rule to obtain
$\Gamma\vdash b:\forall x:\phi'.\phi$.

\ 

\noindent \textbf{Case:}

\

$\infer{\Gamma \vdash a : [a'/x]\phi}{\Gamma \vdash a:\forall x:\phi'.\phi & \Gamma \vdash a':\phi'}$

\ 

\noindent By the induction hypothesis, we have $\Gamma \vdash
b:\forall x:\phi'.\phi$.  We may then reapply this rule to obtain
$\Gamma \vdash b:[a'/x]\phi$.

\ 

\noindent \textbf{Case:}

\

$\infer{\Gamma \vdash \lambda x.a : \Pi x:\phi'.\phi}
      {\Gamma, x:\phi' \vdash a:\phi}$

\ 

\noindent By the induction hypothesis, we have $\Gamma, x:\phi' \vdash
b:\phi$.  We may then reapply this rule to obtain $\Gamma \vdash
\lambda x.b : \Pi x:\phi'.\phi$.

\ 

\noindent \textbf{Case:}

\

$\infer{\Gamma \vdash (a\ a') : [a'/x]\phi}{\Gamma \vdash a:\Pi x:\phi'.\phi & \Gamma\vdash a':\phi'}$

\ 

\noindent Suppose the reduction is from $a\leadsto b$, so we have
$(a\ a')\leadsto(b\ a')$.  Then we apply the induction hypothesis to
the first premise to obtain $\Gamma \vdash b:\Pi x:\phi'.\phi$, and
then reapply this rule to obtain $\Gamma \vdash (b\ a') : [a'/x]\phi$.
Suppose now that the reduction is from $a'\leadsto b'$, so we have
$(a\ a')\leadsto(a\ b')$.  Then we apply the induction hypothesis to
the second premise to obtain $\Gamma\vdash b':\phi'$.  Reapplying this
rule then gives us $\Gamma \vdash (a\ b') : [b'/x]\phi$.  We must now
apply the \texttt{conv} rule (of Figure~\ref{fig:typing}), using as
the first premise the judgment $\Gamma \vdash \join : a' = b'$, which
is derivable since $a'\downarrow b'$ (because $a'\leadsto b'$).  This gives
us the desired result: $\Gamma \vdash (a\ b') : [a'/x]\phi$.

Finally, suppose the reduction is because we have $a \equiv \lambda
x.b$ for some $x$ and $b$, and the application itself is being
$\beta$-reduced.  In this case, we need a lemma in order to limit the
cases arising from inversion on the derivation of $\Gamma\vdash
\lambda x.b:\Pi x:\phi'.\phi$.  We now need this lemma (proof in
Section~\ref{sec:simplinv}):

\begin{lemma}[Simplifying Inversion]
\label{lem:simplinv}
Suppose $\Gamma\vdash a:\phi$ is derivable, where $a$ is an
introduction form (i.e., of the form $\join$, $0$, $(S\ b)$, $\nil$,
$(\cons\ b\ b')$, or $\lambda x.b$), and $\phi$ has the corresponding
form of type (e.g., a $\Pi$-type for a $\lambda$-abstraction).  Then
$\Gamma\vdash a:\phi$ is also derivable by a derivation starting with
the corresponding introduction rule for the form of $a$, using the
same context $\Gamma$, and followed by a sequence of \texttt{(conv)}
inferences.
\end{lemma}

\noindent Using Simplifying Inversion on the derivation $\Gamma\vdash
\lambda x.b:\Pi x:\phi'.\phi$, we may assume this derivation starts
like this:
\[
\infer{\Gamma\vdash\lambda x.b:\Pi x:\psi'.\psi}
      {\Gamma,x:\psi'\vdash b:\psi}
\]
\noindent The derivation then uses a sequence $S$ of \texttt{(conv)}
inferences to end in $\Gamma\vdash \lambda x.b:\Pi x:\phi'.\phi$.  Let
$S^{-1}$ be the sequence which is just the same except that for every
first premise $\Gamma\vdash d:c = c'$ of a \texttt{(conv)}-inference
in $S$, we have a \texttt{(conv)} inference with first premise
$\Gamma\vdash \join:c' = c$, easily derived from $\Gamma\vdash d:c =
c'$.  We now wish to show that the result of substituting our $a'$ for
$x$ in $b$ has the expected type $[a'/x]\phi$.  For this, we must
first apply the sequence $S^{-1}$ to $\Gamma\vdash a':\phi'$.  This
gives us $\Gamma\vdash a':\psi'$.  Now we apply Substitution (proved
in Section~\ref{sec:subst} below):

\begin{lemma}[Substitution]
\label{lem:subst}
If $\Gamma,x:\phi,\Gamma'\vdash a':\phi'$ and $\Gamma\vdash a:\phi$,
then $\Gamma,[a/x]\Gamma'\vdash [a/x]a':[a/x]\phi'$.
\end{lemma}

\noindent This gives us $\Gamma\vdash[a/x]a':[a/x]\psi$.  We may
apply $S$ now to obtain $\Gamma\vdash[a/x]a':[a/x]\phi$. 

\ 

\noindent \textbf{Case:}

\

$\infer{\Gamma\vdash 0:\nat}{\ }$

\

\noindent This case cannot arise since $0$ is a normal form.

\ 

\noindent \textbf{Case:}

\

$\infer{\Gamma\vdash \nil:\langle\vc\ \phi\ 0\rangle}{\ }$

\ 

\noindent This case cannot arise since $\nil$ is a normal form.

\ 

\noindent \textbf{Case:}

\

$\infer{\Gamma\vdash (S\ a):\nat}{\Gamma\vdash a:\nat  }$

\ 

\noindent By the induction hypothesis, we have $\Gamma\vdash b:\nat$,
and we may then reapply this rule to obtain $\Gamma\vdash
(S\ b):\nat$.

\ 

\noindent \textbf{Case:}

\

$\infer{\Gamma\vdash (R_\nat\ a\ a'\ a''):[a''/x]\phi}
      {\begin{array}{l}\Gamma \vdash a'' : \nat \\
       \Gamma \vdash a : [0/x]\phi \\
       \Gamma \vdash a' : \Pi y:\nat. \Pi u : [y/x]\phi. [(S y)/x]\phi
       \end{array}}$

\ 

\noindent If the reduction arises from $a\leadsto b$ or $a'\leadsto b$
or $a''\leadsto b$, then we apply the induction hypothesis to the
corresponding premise and then reapply this typing rule.  If the
reduction arises because $a''\equiv 0$ and the $R_\nat$-term is itself
being reduced to $a$, then we have $\Gamma\vdash a:[a''/x]\phi$ from
the second premise to the rule, and the fact that $a''\equiv
0$. Suppose the reduction arises because $a''\equiv (S\ b)$ and the
$R_\nat$-term is itself being reduced to
$(a'\ b\ (R_\nat\ a\ a'\ b))$.  Applying Simplifying Inversion
(Lemma~\ref{lem:simplinv}), we obtain $\Gamma\vdash b:\nat$.  
So we may apply the $R_\nat$ typing rule
to obtain $\Gamma\vdash(R_\nat\ a\ a'\ b):[b/x]\phi$.  Applying the
application typing rule twice gives us then
$\Gamma\vdash(a'\ b\ (R_\nat\ a\ a'\ b)):[(S\ b)/x]\phi$. This is the
desired typing, since $a''\equiv (S\ b)$.

\ 

\noindent \textbf{Case:}

\

$\infer{\Gamma\vdash (\cons\ a\ a'):\langle \vc\ \phi\ (S\ l)\rangle}
      {\begin{array}{l}\Gamma\vdash a:\phi \\ \Gamma \vdash a':\langle \vc\ \phi\ l\rangle
       \end{array}}$

\ 

\noindent The reduction must arise from $a\leadsto b$ or $a'\leadsto
b$, so we apply the induction hypothesis to the corresponding premise,
and then reapply this typing rule.

\ 

\noindent \textbf{Case:}

\

$\infer{\Gamma\vdash (R_\vc\ a\ a'\ a''):[l/y, a''/x]\phi}
      {\begin{array}{l}\Gamma \vdash a'' : \langle \vc\ \phi'\ l\rangle \\
       \Gamma \vdash a : [0/y,\nil/x]\phi \\
       \Gamma \vdash a' : \Pi z:\phi'. \forall l:\nat. \Pi v :\langle \vc\ \phi'\ l\rangle. \Pi u : [l/y, v/x]\phi. \\
        \ \ \ \ \ \ \ \ \ \ \ \ \  [(S\ l)/y, (\cons\ z\ v)/x]\phi
       \end{array}}$

\ 

\noindent If the reduction arises from $a\leadsto b$ or $a'\leadsto b$
or $a''\leadsto b$, then we apply the induction hypothesis to the
corresponding premise and then reapply this typing rule.  If the
reduction arises because $a''\equiv\nil$ and the $R_\nat$-term is
itself being reduced to $a$, then we have $\Gamma\vdash
a:[0/y,a''/x]\phi$ from the second premise to the rule, and the fact
that $a''\equiv \nil$.  By Simplifying Inversion
(Lemma~\ref{lem:simplinv}), we know there is a derivation of
$\Gamma\vdash\nil:\langle\vc\ \phi'\ l\rangle$ which starts with
$\Gamma\vdash\nil:\langle\vc\ \phi''\ 0\rangle$ and then has a
sequence $S$ of \texttt{(conv)}-inferences.  We may use this same
series $S$ to change $0$ to $l$ in $\Gamma\vdash a:[0/y,a''/x]\phi$,
yielding the desired conclusion.  Finally, suppose the reduction
arises because $a''\equiv(\cons\ b'\ b'')$ for some $b'$ and $b''$,
and the $R_\vc$-term itself is reduced to
$(a'\ b'\ b''\ (R_\vc\ a\ a'\ b''))$.  By Simplifying Inversion again,
we have a derivation of $\Gamma \vdash
(\cons\ b'\ b''):\langle\vc\ \phi'\ l\rangle$ starting from a
$\cons$-introduction deriving $\Gamma\vdash
(\cons\ b'\ b''):\langle\vc\ \phi''\ (S\ l'')$ from premises $\Gamma
b':\phi''$ and $\Gamma b'':\langle\vc\ \phi''\ l''\rangle$; and then
using a sequence $S$ of \texttt{(conv)} inferences.

We may apply the sequence $S$ of \texttt{(conv)} inferences to the
typing for $b''$ to obtain $\Gamma
b'':\langle\vc\ \phi'\ \hat{l}rangle$, for some $\hat{l}$ where
$(S\ \hat{l})\equiv l$.  With this, we can reapply the typing rule for
$R_\vc$ to get
$\Gamma\vdash(R_\vc\ a\ a'\ b''):[\hat{l}/y,b''/x]\phi$.  Using this
and the typing for $b''$ we derived just previously, we can obtain
$\Gamma\vdash(a'\ b'\ b''\ (R_\vc\ a\ a'\ b'')):[(S\ \hat{l})/y,(\cons\ b'\ b'')/x]\phi$.
Since $a\equiv(\cons\ b'\ b'')$ and $(S\ \hat{l})\equiv l$, the type
is equivalent to the desired one.

\subsection{Proof of Simplifying Inversion (Lemma~\ref{lem:simplinv})}
\label{sec:simplinv}

For the proof of this lemma, we begin by simplifying the derivation
$\mathcal{D}$ of $\Gamma\vdash a:\phi$ by applying two transformations
to the maximal path ending with the conclusion $\mathcal{D}$, which is
assigning a type to $a$ (rather than a strict subterm of $a$). First,
we remove all inferences of the following form:
\[
\infer{\Gamma\vdash a:[a'/x]\phi}
      {\infer{\Gamma\vdash a:\forall x:\phi'.\phi}
             {\Gamma, x:\phi'\vdash a:\phi} &
       \Gamma \vdash a':\phi'}
\]
\noindent We may replace this with the result of applying the
Substitution Lemma proved in the next section, in the special case
where $x\not\in\textit{FV}(a)$.  By inspection of the proof of
Substitution, we see that while inferences of the form we are eliminating
can be created, they must be created in a part of $\mathcal{D}$ typing
a strict subterm of $a$.  This is because $a$ is an introduction form.

The second transformation simplifies this inference:
\[
\infer{\Gamma\vdash a:[a'/y][a_2/x]\phi}
      {\infer{\Gamma\vdash a:\forall y:[a_2/x]\phi'.[a_2/x]\phi}
             {\Gamma\vdash \hat{a}:a_1 = a_2 &
              \Gamma\vdash a:\forall y:[a_1/x]\phi'.[a_1/x]\phi} &
       \Gamma\vdash a':[a_2/x]\phi'}
\]
\noindent Observing that the substitutions in question commute,
we reduce this to the following, where $\hat{a}'$ is easily
constructed to show $a_2 = a_1$ from $\hat{a}:a_1 = a_2$:
\[
\infer{\Gamma\vdash a:[a_2/x][a'/y]\phi}
      {\Gamma\vdash \hat{a}:a_1 = a_2 &
       \infer{\Gamma\vdash a:[a_1/x][a'/y]\phi}
             {\Gamma\vdash a:\forall y:[a_1/x]\phi'.[a_1/x]\phi &
              \infer{\Gamma\vdash a':[a_1/x]\phi'}
                    {\Gamma\vdash \hat{a}':a_2 = a_1 & \Gamma\vdash a':[a_2/x]\phi'}}}
\]

\noindent For each of these transformations, the following measure is
strictly decreased in the lexicographic ordering (combining two copies
of the natural number ordering): the pair of the sum of the distances
of occurrences of a \texttt{(conv)} inference on the maximal path
typing $a$; and the number of inferences of the form removed by the
first transformation in that same path.  Note that the \texttt{(conv)}
inference introduced by the second transformation in the topmost
rightmost position show above is not on the maximal path typing $a$
(it is typing $a'$).  Strict decrease of the stated measure implies
that the transformations terminate.  They also preserve the form of
the derived judgment.

We can now prove the lemma by induction on the simplified derivation
$\mathcal{D}$.  It cannot end in a use of \texttt{(spec-abs)}, since
then $\phi$ would be a $\forall$-type (and by assumption it is of the
form corresonding to the introduction form which $a$ is assumed to
have).  It also cannot be a \texttt{(spec-app)} inference, for the
following reason.  Consider the maximal consecutive sequence $S$ of
\texttt{(spec-app)} inferences ending at the conclusion of
$\mathcal{D}$, and typing $a$.  These inferences cannot start with the
conclusion of either a \texttt{(conv)} or a \texttt{(spec-abs)}
inference, since such patterns of inference have been eliminated by
the above transformations.  But these are the only possibilities, since
$a$ is an introduction form.  Therefore, the derivation $\mathcal{D}$
ends in a sequence of \texttt{(conv)} inferences, starting from a use
of the introduction rule for $a$.  Since \texttt{(conv)} does not change
the context, this introduction inference for $a$ uses the same context,
as required by the statement of the lemma.

\subsection{Proof of Substitution (Lemma~\ref{lem:subst})}
\label{sec:subst}

The proof is written using different variable names than the statement
of the lemma, in order to not clash with the variable names in the
typing rules. We prove:

If $\Gamma,y:\psi,\Gamma'\vdash a:\phi$ and $\Gamma\vdash b:\psi$,
then $\Gamma,[b/y]\Gamma'\vdash [b/y]a:[b/y]\phi$.

The proof is by induction on the depth of $\Gamma,y:\psi,\Gamma'\vdash a:\phi$. The cases are:

\ 

\noindent \textbf{Case:}

\

$\infer{\Gamma,y:\psi,\Gamma'\vdash x:\phi}{(\Gamma,y:\psi,\Gamma')(x) \equiv \phi & \Gamma \mathit{Ok}}$

There are three cases: $x \in dom(\Gamma)$, $x \in dom(\Gamma')$, or $x = y$. If $x \in dom(\Gamma)$, then by $\Gamma \mathit{Ok}$ we know $y \not\in FV(\phi)$, so $[b/y]\phi \equiv \phi$ and the conclusion follows by Var. If $x \in dom(\Gamma')$ the conclusion follows directly by Var. In the case $x = y$, we use the second assumption together with a weakening lemma:

\begin{lemma}[Weakening]
\label{lem:weak}
If $\Gamma \vdash a : \phi$, $\textit{dom}(\Gamma) \subset \textit{dom}(\Gamma')$, and $\Gamma' \mathit{Ok}$, then $\Gamma,\Gamma' \vdash a : \phi$.
\end{lemma}
\begin{proof}
Induction on $\Gamma \vdash a : \phi$. The only interesting cases are for Abs and Spec-Abs. There we have $\Gamma, x:\phi' \vdash a : \phi$ by assumption; by regularity, we get $\Gamma, x:\phi' \mathit{Ok}$, so then $\Gamma', x:\phi' \mathit{Ok}$ and the conclusion follows by IH.
\end{proof}

\begin{lemma}[Free variables of typable terms]
\label{lem:vars}
If $\Gamma\vdash a:\phi$, then $\textit{FV}(a)\subset\textit{dom}(\Gamma)$.
\end{lemma}
\noindent The proof is a straightforward induction on the typing
derivation.

We also need to note that by Lemma~\ref{lem:vars} $\mathit{FV}(b) \subset \mathit{dom}(\Gamma)$, and the substitution preserves well-scoping of contexts:

\begin{lemma}
If $\Gamma, y:\psi, \Gamma' \mathit{Ok}$ and $\mathit{FV}(b)\subset \mathit{dom}(\Gamma)$, then $\Gamma, [b/y]\Gamma' \mathit{Ok}$.
\end{lemma}
\begin{proof}
For each variable $z$ in $\Gamma'$, if the entire context looks like  $\Gamma, y:\psi, \Gamma'', z:\phi, \Gamma'''$ we know that $\mathit{FV}([b/y]\phi) \subset \mathit{FV}(\phi)\cup\mathit{FV}(b)-\{y\}$ and $\mathit{dom}(\Gamma, [b/y]\Gamma'') = \mathit{dom}(\Gamma, y:\psi, \Gamma'') - \{y\}$, so $[b/y]\phi$ is still well-scoped.
\end{proof}

\ 

\noindent \textbf{Case:}

\

$\infer{\Gamma,y:\psi,\Gamma'\vdash x:\phi\vdash \join : a = a'}{a \downarrow a'}$

\ 

\noindent Immediate by the fact that substitution preserves joinability.

\ 

\noindent \textbf{Case:}

\

$\infer{\Gamma,y:\psi,\Gamma'\vdash a:[a''/x]\phi}{\Gamma,y:\psi,\Gamma' \vdash a''': a' = a'' & \Gamma,y:\psi,\Gamma'\vdash a:[a'/x]\phi & x\not\in\textit{dom}(\Gamma,y:\psi,\Gamma')}$

\ 

\noindent First, rename $x$ in $\phi$ so that $x \not\in \mathit{FV}(b)$.

By IH we get $\Gamma, [b/y]\Gamma' \vdash [b/y]a''' : [b/y]a' = [b/y]a''$ and $\Gamma, [b/y]\Gamma' \vdash [b/y]a : [b/y][a'/x]\phi \equiv [[b/y]a'/x][b/y]\phi$. Now by \texttt{(conv)} we conclude $\Gamma, [b/y]\Gamma' \vdash [b/y]a : [[b/y]a''/x]\phi \equiv [b/y][a''/x]\phi$ as required.

\ 

\noindent \textbf{Case:}

\

$\infer{\Gamma,y:\psi,\Gamma' \vdash a : \forall x:\phi'.\phi}{\Gamma,y:\psi,\Gamma', x:\phi' \vdash a:\phi & x\not\in\textit{FV}(a)}$

\ 

\noindent First, rename $x$ in the derivation of $\Gamma,y:\psi,\Gamma', x:\phi' \vdash a:\phi$ so that $x \not\in \mathit{FV}(a) \cup \mathit{FV}(b)$. This can be done without changing the depth.

By IH we get $\Gamma,[b/y]\Gamma', x:[b/y]\phi' \vdash [b/y]a : [b/y]\phi$. The conclusion follows by Spec-Abs.

\ 

\noindent \textbf{Case:}

\

$\infer{\Gamma,y:\psi,\Gamma' \vdash a : [a'/x]\phi}{ \Gamma,y:\psi,\Gamma' \vdash a:\forall x:\phi'.\phi & \Gamma,y:\psi,\Gamma' \vdash a':\phi'}$

\ 

\noindent Pick $x \neq y$ and $x \not\in \mathit{FV}(b)$. 

By IH we get $\Gamma, [b/y]\Gamma' \vdash [b/y]a : [b/y]\forall x:\phi'.\phi \equiv \forall x:[b/y]\phi'.[b/y]\phi$ and $\Gamma, [b/y]\Gamma' \vdash [b/y]a' : [b/y]\phi'$. Then by Spec-App,  $\Gamma, [b/y]\Gamma' \vdash [b/y]a : [[b/y]a'/x][b/y]\phi \equiv [b/y][a'/x]\phi$ as required.

\noindent \textbf{Case:}

\

$\infer{\Gamma,y:\psi,\Gamma',x:\phi' \vdash \lambda x.a : \Pi x:\phi'.\phi}
      {\Gamma,y:\psi,\Gamma' \vdash a:\phi}$

\ 

\noindent Similar to Spec-abs.

\

\noindent \textbf{Case:}

\

$\infer{\Gamma,y:\psi,\Gamma' \vdash (a\ a') : [a'/x]\phi}{\Gamma,y:\psi,\Gamma'\vdash a:\Pi x:\phi'.\phi & \Gamma,y:\psi,\Gamma'\vdash a':\phi'}$

\ 

\noindent Similar to Spec-App.

\

\noindent \textbf{Case:}

\

$\infer{\Gamma\vdash 0:\nat}{\ }$

\

\noindent Immediate.

\ 

\noindent \textbf{Case:}

\

$\infer{\Gamma\vdash \nil:\langle\vc\ \phi\ 0\rangle}{\ }$

\ 

\noindent Immediate.

\ 

\noindent \textbf{Case:}

\

$\infer{\Gamma\vdash (S\ a):\nat}{\Gamma\vdash a:\nat  }$

\ 

\noindent Immediate by IH.

\ 

\noindent \textbf{Case:}

\

$\infer{\Gamma,y_0:\psi,\Gamma'\vdash (R_\nat\ a\ a'\ a''):[a''/x]\phi}
      {\begin{array}{l}\Gamma,y_0:\psi,\Gamma' \vdash a'' : \nat \\
       \Gamma,y_0:\psi,\Gamma'\vdash a : [0/x]\phi \\
       \Gamma,y_0:\psi,\Gamma'\vdash a' : \Pi y:\nat. \Pi u : [y/x]\phi. [(S y)/x]\phi
       \end{array}}$

\ 

\noindent First rename $x$ in $\phi$ so that $x \not\in \mathit{FV}(b)\cup\{z,y\}$.

IH gives \[ \Gamma,[b/y_0]\Gamma'\vdash a : [b/y_0][0/x]\phi \equiv [0/x][b/y_0]\phi \] and \[ \Gamma,[b/y_0]\Gamma'\vdash a' : [b/y_0] \Pi y:\nat. \Pi u : [y/x]\phi. [(S y)/x]\phi \equiv  \Pi y:\nat. \Pi u : [y/x][b/y_0]\phi. [(S y)/x][b/y_0]\phi
\] Then by Rnat, \[
  \Gamma,[b/y_0]\Gamma'\vdash (R_\nat\ a\ a'\ a''):[[b/y_0]a''/x][b/y_0]\phi \equiv [b/y_0][a''/x]\phi
\]

\ 

\noindent \textbf{Case:}

\

$\infer{\Gamma,y:\psi,\Gamma'\vdash (\cons\ a\ a'):\langle \vc\ \phi\ (S\ l)\rangle}
      {\begin{array}{l}\Gamma,y:\psi,\Gamma'\vdash a:\phi \\ \Gamma,y:\psi,\Gamma' \vdash a':\langle \vc\ \phi\ l\rangle
       \end{array}}$

\ 

\noindent Immediate by IH.

\ 

\noindent \textbf{Case:}

\

$\infer{\Gamma\vdash (R_\vc\ a\ a'\ a''):[l/y, a''/x]\phi}
      {\begin{array}{l}\Gamma \vdash a'' : \langle \vc\ \phi'\ l\rangle \\
       \Gamma \vdash a : [0/y,\nil/x]\phi \\
       \Gamma \vdash a' : \Pi z:\phi'. \forall l:\nat. \Pi v :\langle \vc\ \phi'\ l\rangle. \Pi u : [l/y, v/x]\phi. \\
        \ \ \ \ \ \ \ \ \ \ \ \ \  [(S\ l)/y, (\cons\ z\ v)/x]\phi
       \end{array}}$

\ 

\noindent Similar to Rnat case.

\section{Proof of Progress (Theorem~\ref{thm:progress})}
\label{sec:progress}

The proof is by induction on $\Gamma \vdash a : \phi$. 

\ 

\noindent \textbf{Case:}

\

$\infer{\Gamma\vdash x:\phi}{\Gamma(x) \equiv \phi}$

\ 

\noindent Impossible by $dom(\Gamma) \cap FV(v) = \emptyset$.

\ 

\noindent \textbf{Case:}

\

$\infer{\Gamma\vdash \join : a = a'}{a \downarrow a'}$

\ 

\noindent $\join$ is a value, as required.

\ 

\noindent \textbf{Case:}

\

$\infer{\Gamma\vdash a:[a''/x]\phi}{\Gamma \vdash a''': a' = a'' & \Gamma \vdash a:[a'/x]\phi & x\not\in\textit{dom}(\Gamma)}$

\ 

\noindent Directly by the IH for $\Gamma \vdash a:[a'/x]\phi$.

\ 

\noindent \textbf{Case:}

\

$\infer{\Gamma \vdash a : \forall x:\phi'.\phi}{\Gamma, x:\phi' \vdash a:\phi & x\not\in\textit{FV}(a)}$

\ 

\noindent The condition $x\not\in\textit{FV}(a)$ ensures that $dom(\Gamma,x:\phi') \cap \textit{FV}(a) = \emptyset$, so we can apply the IH for $\Gamma, x:\phi' \vdash a:\phi$.

\ 

\noindent \textbf{Case:}

\

$\infer{\Gamma \vdash a : [a'/x]\phi}{\Gamma \vdash a:\forall x:\phi'.\phi & \Gamma \vdash a':\phi'}$

\ 

\noindent Directly by the IH for $\Gamma \vdash a:\forall x:\phi'.\phi$.

\noindent \textbf{Case:}

\

$\infer{\Gamma \vdash \lambda x.a : \Pi x:\phi'.\phi}
      {\Gamma, x:\phi' \vdash a:\phi}$

\ 

\noindent $\lambda x.a$ is a value as required.

\noindent \textbf{Case:}

\

$\infer{\Gamma \vdash (a\ a') : [a'/x]\phi}{\Gamma \vdash a:\Pi x:\phi'.\phi & \Gamma\vdash a':\phi'}$

\ 

\noindent By the IH for $\Gamma \vdash a:\Pi x:\phi'.\phi$, we know $a$ either steps or is a value. If $a$ steps, the entire expression $(a\ a')$ steps also, by $\leadsto$-congruence. If $a$ is a value, by Lemma~\ref{lem:canonical} $a=\lambda a.a_o$, so $(a\ a')$ steps by $\beta$.

\noindent \textbf{Case:}

\

$\infer{\Gamma\vdash 0:\nat}{\ }$

\

\noindent $0$ is a value as required.

\ 

\noindent \textbf{Case:}

\

$\infer{\Gamma\vdash \nil:\langle\vc\ \phi\ 0\rangle}{\ }$

\ 

\noindent $\nil$ is a value as required.

\ 

\noindent \textbf{Case:}

\

$\infer{\Gamma\vdash (S\ a):\nat}{\Gamma\vdash a:\nat  }$

\ 

\noindent By the IH for $\Gamma\vdash a:\nat$, we know $a$ either steps or is a value; accordingly $S a$ steps or is a value.

\ 

\noindent \textbf{Case:}

\

$\infer{\Gamma\vdash (R_\nat\ a\ a'\ a''):[a''/x]\phi}
      {\begin{array}{l}\Gamma \vdash a'' : \nat \\
       \Gamma \vdash a : [0/x]\phi \\
       \Gamma \vdash a' : \Pi y:\nat. \Pi u : [y/x]\phi. [(S y)/x]\phi
       \end{array}}$

\ 

\noindent By the IH for $\Gamma \vdash a'' : \nat$ and Lemma~\ref{lem:canonical}, we know that either $a''$ steps or $a'' = 0$ or $a'' = S\ v$. Then $(R_\nat\ a\ a'\ a'')$ steps by congruence or one of the two reduction rules, respectively.

\ 

\noindent \textbf{Case:}

\

$\infer{\Gamma\vdash (\cons\ a\ a'):\langle \vc\ \phi\ (S\ l)\rangle}
      {\begin{array}{l}\Gamma\vdash a:\phi \\ \Gamma \vdash a':\langle \vc\ \phi\ l\rangle
       \end{array}}$

\ 

\noindent By the IH, $a$ and $a'$ either step or are values; accordingly $(\cons\ a\ a')$ steps by congruence or is a value.

\ 

\noindent \textbf{Case:}

\

$\infer{\Gamma\vdash (R_\vc\ a\ a'\ a''):[l/y, a''/x]\phi}
      {\begin{array}{l}\Gamma \vdash a'' : \langle \vc\ \phi'\ l\rangle \\
       \Gamma \vdash a : [0/y,\nil/x]\phi \\
       \Gamma \vdash a' : \Pi z:\phi'. \forall l:\nat. \Pi v :\langle \vc\ \phi'\ l\rangle. \Pi u : [l/y, v/x]\phi. \\
        \ \ \ \ \ \ \ \ \ \ \ \ \  [(S\ l)/y, (\cons\ z\ v)/x]\phi
       \end{array}}$

\ 

\noindent By the IH for $\Gamma \vdash a'' : \langle \vc\ \phi'\ l\rangle$ and Lemma~\ref{lem:canonical}, we know that either $a''$ steps or $a'' = \nil$ or $a'' = (\cons\ v\ v')$. Then $(R_\vc\ a\ a'\ a'')$ steps by congruence or by one of the two reduction rules, respectively.

\subsection{Proof of Canonical Forms (Lemma~\ref{lem:canonical})}
\label{sec:canonical}

Write $\forall \vec{x}.\phi$ for $\forall x_1:\phi_1'. \forall x_2:\phi'_2. \dots \forall x_n:\phi_n'.\phi$.
\begin{lemma}[Canonical Forms]
\label{lem:canonical}
If $\textit{dom}(\Gamma)\cap FV(v) = \emptyset$, then
\begin{itemize}
\item if $\Gamma \vdash v : \forall \vec{x}. \nat$, then $v = 0$ or $v = S\ v$.
\item if $\Gamma \vdash v : \forall \vec{x}. \langle\vc\ \phi\ l\rangle$, then $v = \nil$ or $v = \cons\ v'\ v''$.
\item if $\Gamma \vdash v : \forall \vec{x}. \Pi x:\phi'.\phi$, then $v = \lambda x.a$
\item if $\Gamma \vdash v : \forall \vec{x}. a_1 = a_2$, then $v = \join$.
\end{itemize}
\end{lemma}

\noindent \textbf{Proof:}

Induction on the typing judgment. The cases are 
\begin{itemize}
\item \textbf{var} Impossible by $dom(\Gamma) \cap FV(v) = \emptyset$.
\item \textbf{conv}
\

$\infer{\Gamma\vdash a:[a''/x]\phi}{\Gamma \vdash a''': a' = a'' & \Gamma \vdash a:[a'/x]\phi & x\not\in\textit{dom}(\Gamma)}$

\ 

The types $[a''/x]\phi$ and $[a'/x]\phi$ have the same top-level structure, so the IH applies.

\item \textbf{spec-abs}
\

$\infer{\Gamma \vdash a : \forall x:\phi'.\phi}{\Gamma, x:\phi' \vdash a:\phi & x\not\in\textit{FV}(a)}$

\ 

The condition $x\not\in \textit{FV}(a)$ ensures that $dom(\Gamma, x:\phi') \cup FV(v) = \emptyset$. Also, the type $\phi$ still has the required form. So the IH applies.

\item \textbf{spec-app}
\

$\infer{\Gamma \vdash a : [a'/x]\phi}{\Gamma \vdash a:\forall x:\phi'.\phi & \Gamma \vdash a':\phi'}$

\ 

If the type $[a'/x]\phi$ has the required form, then $\forall x:\phi'.\phi$ has required form also, so the IH applies.

\item \textbf{join,abs,zero,nil,succ,cons} The value in the conclusion has the required form.
\item \textbf{app, Rnat, Rvec} The term in the conclusion is not a value.

\end{itemize}

\section{Proof of Critical Properties (Section~\ref{sec:critprop})}

\subsection{More basic notation}

We will define a term context to be a term $a^*$ with a designated free
variable $*$, which may be instantiated in a capture-avoiding way by a
term $a'$ using the notation $a^*[a']$.  

Also, if $S$ is a set of terms, then we will allow ourselves to write
a term which has $S$ inserted for the hole of some context $a^*$.  For
example, we may write $\lambda x. \{ x, (x\ x) \}$ (here $a^* =
\lambda x.*$).  The meaning of this notation is the set of terms
$\{a^*[a'] | a' \in S\}$.  So in the example: $\{ \lambda
x.x,\ \lambda x.(x\ x) \}$.  

Finally, if $a\in\SN$, we define $\nu(a)$ to be some bound on the
lengths of the reduction sequences from $a$.

\subsection{Preliminary observation}

We will not explicitly prove strong normalization or typability in the
cases for \textbf{R-Join} below.  This is because \textbf{R-Join}
assumes $\Gamma\vdash a':a_1 = a_2$ for some $a'$, so we will always
be able to show $\Gamma\vdash a:[a_2/x]\phi$ from $\Gamma\vdash
a:[a_1/x]\phi$ using \texttt{(conv)}.  Similarly, we will always have
$a\in\SN$, since it follows from $a\in\interp{[a_1/x]\phi}_\Gamma$.

\subsection{Critical properties for $\nat$}

\textbf{R-Pres} holds using Type Preservation
(Theorem~\ref{thm:tppres}), and the fact that $a\in\SN \To
\next(a)\subset\SN$.  For \textbf{R-Prog}, we have $\next(a)\subset\SN
\To a\in\SN$, and $\Gamma\vdash a:\nat$ by assumption (of
\textbf{R-Prog}). We have \textbf{R-Join} because all instances
$[a/x]\nat \equiv \nat$ have the same (trivial) interpretation.

\subsection{Critical properties for $\langle\vc\ \phi\ l\rangle$}

\subsubsection{Proof of \textbf{R-Pres}}

Suppose $a\in\interp{\langle\vc\ \phi\ l\rangle}_\Gamma$, and
$a\leadsto a'$.  We have $a'\in\SN$ from $a\in\SN$, and $\Gamma\vdash
a':\phi$ by Type Preservation.  To prove the second conjunct for
$\interp{\cdot}$ at $\vc$-type -- namely, $(a'\leadsto^* \nil\ \To\ l
\sim_\Gamma 0)$ -- assume $a'\leadsto^* \nil$.  From this and the
assumption that $a\leadsto a'$, we have $a\leadsto^* \nil$.  So we can
use the corresponding conjunct for $a$ to conclude $l\sim_\Gamma 0$,
as required.  Similar reasoning applies for the third conjunct.

\subsubsection{Proof of \textbf{R-Prog}}

Suppose $\next(a)\subset \interp{\langle\vc\ \phi\ l\rangle}_\Gamma$,
with $a$ neutral and $\Gamma\vdash a:\phi$.  We have $a\in\SN$ for the
same reason as for the $\nat$ case above.  It suffices now to show the
two conjuncts of the clause of $\interp{\cdot}$ for $\vc$-types, for
$a$.  For the first, assume $a\leadsto^*\nil$.  Now since $a$ is
neutral, we cannot have $a \equiv \nil$.  So consider arbitrary
$a'\in\next(a)$.  From $a\leadsto^*\nil$ and $a\leadsto a'$, we obtain
$a'\leadsto^*\nil$ by confluence.  We may then apply the corresponding
second conjunct for $a'$ to obtain the desired result for this
conjunct.  The second conjunct follows by similar reasoning.

\subsubsection{Proof of \textbf{R-Join}}

Suppose that $a_1\sim_\Gamma a_2$, and consider arbitrary
$a\in\interp{[a_1/x]\langle\vc\ \phi\ l\rangle}_\Gamma$.  We must show
$a\in\interp{[a_2/x]\langle\vc\ \phi\ l\rangle}_\Gamma$.  It suffices
to prove the two conjuncts for $\interp{\cdot}$ at the type involving
$a_2$, assuming them for the type involving $a_1$.  For the first
conjunct, assume $a\leadsto^*\nil$.  We must show that for an
arbitrary $\sigma\in\interp{\Gamma}$, we have $\sigma
([a_2/x]l)\downarrow 0$.  From the second conjunct for the $a_1$-type,
we have $\sigma ([a_1/x]l)\downarrow 0$.  Instantiate our first
assumption about $a_1$ and $a_2$, to obtain $\sigma a_1\downarrow
\sigma a_2$.  Now we use the fact that joinability is closed under
substitution of joinable terms (proof omitted), to obtain the desired
result.  Note that joinability is not closed under substitution of
joinable terms for more specialized reduction strategies, such as
call-by-value or call-by-name.

For the second conjunct, assume $a\leadsto^*(\cons\ a'\ a'')$.  From
the corresponding second conjunct for the $a_1$-type, we obtain the
following for some $l'$:
\begin{itemize}
\item $a'\in\interp{[a_1/x]\phi}_\Gamma$
\item $a''\in\interp{\langle\vc\ [a_1/x]\phi\ l'\rangle}_\Gamma$
\item $\forall \sigma\in\interp{\Gamma}.\ \sigma ([a_1/x]l) \downarrow (S\ \sigma l')$ 
\end{itemize}
\noindent We use \textbf{R-Join} on the first formula to derive the
similar statement involving $a_2$. The measure
$(|\Gamma|,d(\phi),l(a))$ decreases, because the depth of the type
decreases (and $|\Gamma|$ is unchanged).  Now let $z$ be a variable
not in $\textit{dom}(\Gamma)$ and not free in $\phi$, $a_1$, or $l'$.
Then the second formula is equivalent to
$a''\in\interp{[a_1/z]\langle\vc\ [z/x]\phi\ l'\rangle}_\Gamma$, and
by \textbf{R-Join} (where the measure decreases because the depth of
the type is the same, but the quantity given by $l(\cdot)$ has
decreased), we have
$a''\in\interp{[a_2/z]\langle\vc\ [z/x]\phi\ l'\rangle}_\Gamma$, which
is equivalent to the required
$a''\in\interp{\langle\vc\ [a_2/x]\phi\ l'\rangle}_\Gamma$.
Instantiating the third formula with an arbitrary
$\sigma\in\interp{\Gamma}$, we have $\sigma ([a_1/x]l) \downarrow
(S\ l')$.  We appeal as above to the closure of joinability under
substitution of joinable terms, to obtain $\sigma ([a_2/x]l)
\downarrow (S\ l')$.

\subsection{Critical properties for $\Pi x:\phi'.\phi$}

\subsubsection{Proof of \textbf{R-Pres}}

Assume $a\in\interp{\Pi x:\phi'.\phi}_\Gamma$, and consider an
arbitrary $a'\in\interp{\phi'}^+_\Gamma$.  By definition of
$\interp{\cdot}$, we have $(a\ a')\in\interp{[a'/x]\phi}_\Gamma$.  We
also have
\begin{equation}
(\next(a)\ a') \ \subset\ \next(a\ a') \label{eq7}
\end{equation}

\noindent By \textbf{R-Pres} at type $[a'/x]\phi$ (where we have
$d([a'/x]\phi) < d(\Pi x:\phi'.\phi)$, and so can apply the induction
hypothesis), we obtain
$(\next(a\ a'))\subset\interp{[a'/x]\phi}_\Gamma$.  By (\ref{eq7}),
this implies $(\next(a)\ a')\subset\interp{[a'/x]\phi}_\Gamma$.  We
conclude this for all $a'\in\interp{\phi'}^+_\Gamma$.  Then by the
definition of $\interp{\cdot}$, we obtain the desired
$\next(a)\subset\interp{\Pi x:\phi'.\phi}_\Gamma$, using also Type
Preservation and the fact that $\next(a)\subset\SN$ (since $a\in\SN$).

\subsubsection{Proof of \textbf{R-Prog}}

Suppose $a$ is neutral with $\Gamma\vdash a:\Pi x:\phi'.\phi$.  By
assumption, we have
\begin{equation}
\next(a) \subset \interp{\Pi x:\phi'.\phi}_\Gamma \label{eq8}
\end{equation}

\noindent It suffices, by the definition of $\interp{\cdot}$, to show
that $a\in\SN$ and for all $a'\in\interp{\phi'}^+_\Gamma$,
$(a\ a')\in\interp{[a'/x]\phi}_\Gamma$. We have $a\in\SN$ from
$\next(a)\subset\SN$, so we focus on the latter property.  Consider
arbitrary $a'\in\interp{\phi'}^+_\Gamma$.  Since $a$ is neutral,
$(a\ a')$ cannot be a $\beta$-redex.  Since
$a'\in\interp{\phi'}_\Gamma$, we have $a'\in\SN$ by \textbf{R-SN} at
type $\phi'$ (where $d(\phi') < d(\Pi x:\phi'.\phi)$, so the induction
hypothesis may be applied).  So we may reason by inner induction on
the number $\nu(a')$ to prove that for all
$a'\in\interp{\phi'}^+_\Gamma$, we have
$(a\ a')\in\interp{[a'/x]\phi}_\Gamma$.  By \textbf{R-Prog} at type
$[a'/x]\phi$ (where $d([a'/x]\phi)<d(\Pi x:\phi'.\phi)$, so the
induction hypothesis may be applied), it suffices to prove
$\next(a\ a')\subset\interp{[a'/x]\phi}_\Gamma$, since the term in
question is neutral and since we have $\Gamma \vdash
(a\ a'):[a'/x]\phi$.  The possibilities for reduction are summarized
by:
\begin{equation*}
\next(a\ a')\ \subset\ (\next(a)\ a')\ \cup\ (a\ \next(a'))
\end{equation*}

\noindent We have $(\next(a)\ a')\in\interp{[a'/x]\phi}_\Gamma$ from
(\ref{eq8}), by the definition of $\interp{\cdot}$.  For reducibility
of the second set, consider arbitrary $a''\in\next(a')$.  By our inner
induction hypothesis, which we may apply because
$a''\in\interp{\phi'}_\Gamma$ by \textbf{R-Pres} at type $\phi'$ (with
smaller depth), we have $(a\ a'')\in\interp{[a''/x]\phi}_\Gamma$.  Now
we may apply \textbf{R-Join} at type $\phi$ (with smaller depth),
using the obvious fact that $a' \leadsto a''$ implies the facts $a'
\sim_\Gamma a''$ and $\Gamma\vdash \join:a' = a''$ (required by
\textbf{R-Join}).  This yields
$(a\ a'')\in\interp{[a'/x]\phi}_\Gamma$, as required by our inner
induction.

\subsubsection{Proof of \textbf{R-Join}}

Suppose that $a_1\sim_\Gamma a_2$, and consider arbitrary
$a\in\interp{[a_1/x]\Pi y:\phi'.\phi}_\Gamma$.  We must show
$a\in\interp{[a_2/x]\Pi y:\phi'.\phi}_\Gamma$.  It suffices to show
$(a\ a')\in\interp{[a'/y][a_2/x]\phi}_\Gamma$ for an arbitrary
$a'\in\interp{[a_2/x]\phi'}^+_\Gamma$.  We now wish to use
\textbf{R-Join} at type $\phi'$ (with smaller depth), with the
symmetric equality $a_2\sim_\Gamma a_1$.  Symmetry of $\sim_\Gamma$ is
direct from its definition.

Using \textbf{R-Join} in this way with $a_2\sim_\Gamma a_1$, we obtain
$a'\in\interp{[a_1/x]\phi'}_\Gamma$.  We must further obtain
$a'\in\interp{[a_1/x]\phi'}^+_\Gamma$.  So consider arbitrary
$\sigma\in\interp{\Gamma}$.  From closability of $a'$ at the type
involving $a_2$, we have $\sigma a'\in\interp{\sigma ([a_2/x]\phi')}$.
We must show $\sigma a'\in\interp{\sigma ([a_1/x]\phi')}$.  If
$\Gamma$ is empty, this formula is equivalent to
$a'\in\interp{[a_1/x]\phi'}$, which we already have.  So suppose
$\Gamma$ is not empty.  Then the formula is equivalent to $\sigma
a'\in\interp{[\sigma a_1/x](\sigma \phi')}$, since
$x\not\in\textit{ran}(\sigma)$.  Notice that from our assumption that
$a_1\sim_\Gamma a_2$, we obtain $\sigma a_1 \sim \sigma a_2$.  We may
now use \textbf{R-Join}, where the length of the context has
decreased, to conclude $\sigma a'\in\interp{[\sigma a_1/x](\sigma
  \phi')}$ from $\sigma a'\in\interp{[\sigma a_2/x](\sigma \phi')}$.

Since we have obtained $a'\in\interp{[a_1/x]\phi'}^+_\Gamma$, we now get
$(a\ a')\in\interp{[a'/y][a_1/x]\phi}_\Gamma$ by the assumption above
of reducibility of $a$.  Applying Lemma~\ref{lem:vars} to the fact
that $\Gamma\vdash a':[a_1/x]\phi'$ (which we have from
$a'\in\interp{[a_1/x]\phi'}_\Gamma$), and using the assumption that
$x\not\in\textit{dom}(\Gamma)$, we obtain $x\not\in\textit{FV}(a')$.
Since $y$ is locally scoped, we may also assume that
$y\not\in\textit{FV}(a_1)$ and $y\not\in\textit{FV}(a_2)$. This tells
us that $[a'/y][a_1/x]\phi = [a_1/x][a'/y]\phi$ and also
$[a'/y][a_2/x]\phi = [a_2/x][a'/y]\phi$.  Using the first of these, we
may conclude $(a\ a')\in\interp{[a_1/x][a'/y]\phi}_\Gamma$ from the
similar fact we had just above.  Using the second of these
commutations of substitutions, and also \textbf{R-Join} at type
$[a'/y]\phi$ (of smaller depth), we can conclude
$(a\ a')\in\interp{[a'/y][a_2/x]\phi}_\Gamma$, as required.

\subsection{Critical properties for $\forall x:\phi'.\phi$}

The proofs here are simpler versions (particularly for
\textbf{R-Prog}) of those for the previous case.

\subsubsection{Proof of \textbf{R-Pres}}

Assume $a\in\interp{\forall x:\phi'.\phi}_\Gamma$, and consider an
arbitrary $a'\in\interp{\phi'}^+_\Gamma$.  By \textbf{R-Pres} at type
$[a'/x]\phi$ (with smaller depth), we obtain
$\next(a)\subset\interp{[a'/x]\phi}_\Gamma$.  We conclude this for all
$a'\in\interp{\phi'}^+_\Gamma$.  Then by the definition of
$\interp{\cdot}$, we get the required $\next(a)\subset\interp{\forall
  x:\phi'.\phi}_\Gamma$, using also Type Preservation and the fact
$\next(a)\subset\SN$ (from $a\in\SN$).

\subsubsection{Proof of \textbf{R-Prog}}

Suppose $a$ is neutral with $\Gamma\vdash a:\phi$.  By assumption, we have
\[
\next(a) \subset \interp{\forall x:\phi'.\phi}_\Gamma 
\]

\noindent It suffices, by the definition of $\interp{\cdot}$, to show
that $a\in\SN$ and for all $a'\in\interp{\phi'}^+_\Gamma$,
$a\in\interp{[a'/x]\phi}_\Gamma$. We have $a\in\SN$ from
$\next(a)\subset\SN$, so we focus on the latter property.  Consider
arbitrary $a'\in\interp{\phi'}^+_\Gamma$.  By the definition of
$\interp{\cdot}$ at $\forall$-type and our above assumption, we have
$\next(a)\subset\interp{[a'/x]\phi}_\Gamma$.  So by \textbf{R-Prog} at
type $[a'/x]\phi$ (with smaller depth), we have
$a\in\interp{[a'/x]\phi}_\Gamma$, as required.

\subsubsection{Proof of \textbf{R-Join}}

Suppose that $a_1\sim_\Gamma a_2$, and consider arbitrary
$a\in\interp{[a_1/x]\forall y:\phi'.\phi}_\Gamma$.  We must show
$a\in\interp{[a_2/x]\forall y:\phi'.\phi}_\Gamma$.  It suffices to
show $a\in\interp{[a'/y][a_2/x]\phi}_\Gamma$ for an arbitrary
$a'\in\interp{[a_2/x]\phi'}^+_\Gamma$.  By \textbf{R-Join} at type
$\phi'$ (with smaller depth), and using the symmetric version of our
assumption as above, we have $a'\in\interp{[a_1/x]\phi'}_\Gamma$.  We
further obtain $a'\in\interp{[a_1/x]\phi'}^+_\Gamma$ as in the case
for \textbf{R-Prog} for $\Pi$-types.  So we get
$a\in\interp{[a'/y][a_1/x]\phi}_\Gamma$ by the assumption of
reducibility of $a$.  By similar reasoning as above, we may permute
the substitutions in question.  So we may apply \textbf{R-Join} at
type $[a'/y]\phi$ (of smaller depth) to conclude
$a\in\interp{[a'/y][a_2/x]\phi}_\Gamma$, as required.

\subsection{Critical properties for $a_1 = a_2$}

\subsubsection{Proof of \textbf{R-Pres}}

Consider arbitrary $b$ with $a\leadsto b$.  We
have $b\in\SN$ from $a\in\SN$, and $\Gamma\vdash b: a_1 = a_2$ by Type
Preservation.  Now suppose $b\leadsto^* \join$.  Then we have
$a\leadsto^* \join$ and obtain $a_1\sim_\Gamma a_2$ from
$a\in\interp{a_1 = a_2}_\Gamma$.  

\subsubsection{Proof of \textbf{R-Prog}}

We have $a\in\SN$ from $\next(a)\subset\SN$ as in other cases above.
Suppose that $a\leadsto^* \join$.  Since $a$ is neutral, we cannot
have $a\equiv \join$.  So we must have $a\leadsto b$.  Then we get
$b\leadsto^* \join$ by confluence, and we can use the assumption that
$b\in\interp{a_1 = a_2}_\Gamma$ to obtain $a_1 \sim_\Gamma a_2$ as
required.

\subsubsection{Proof of \textbf{R-Join}}

Assume $a_1'\sim_\Gamma a_2'$, and assume $a\leadsto^*\join$.  Then we
have $[a_1'/x]a_1\sim_\Gamma [a_1'/x]a_2$ from $a\in\interp{[a_1'x]a_1
  = [a_1'/x]a_2}_\Gamma$.  Consider arbitrary
$\sigma\in\interp{\Gamma}$.  Instantiating our two assumptions of
joinability under all ground instances with this $\sigma$, we obtain:
\begin{itemize}
\item $\sigma a_1'\downarrow \sigma a_2'$
\item $(\sigma ([a_1'/x]a_1))\downarrow (\sigma ([a_1'/x]a_2))$
\end{itemize}
\noindent The desired result (namely, $(\sigma
([a_2'/x]a_1))\downarrow (\sigma ([a_2'/x]a_2))$) now follows from 
closure of joinability under substitution of joinable terms.

\section{Proof of Soundness (Theorem~\ref{thm:soundness})}

\subsection{The Closability Lemma}

We have carefully crafted our notions of closable terms and closable
substitutions to allow the following two lemmas to be proved.  The
first expresses the basic desired property of closable substitutions,
and the second shows that under the conditions of the Soundness
Theorem, the term $\sigma a$ is closable which Soundness tells us is
in the interpretation of $\sigma\phi$ with context $\Delta$.

\begin{lemma}[Composing Substitutions]
\label{lem:csubst}
Suppose $\sigma\in\interp{\Gamma}_{\Delta}$ and
$\sigma'\in\interp{\Delta}$.  Then
$\sigma'\circ\sigma\in\interp{\Gamma}$.
\end{lemma}

\noindent The proof is by induction on the
structure of the derivation of
$\sigma\in\interp{\Gamma}_{\Delta}$.  The base case
holds trivially, noting that $\sigma'\circ\emptyset = \emptyset$.  For
the step case, we have
\[
\infer{\sigma''\cup\{(x,a)\}\in\interp{\Gamma',x:\phi}_{\Delta}}
      {a\in\interp{\sigma''\phi}^+_{\Delta} & \sigma'' \in\interp{\Gamma'}_{\Delta}}
\]
\noindent Now we obtain
$\sigma'(\sigma''a)\in\interp{\sigma'(\sigma''\phi)}$, by the
definition of closability of $a$.  This implies
$\sigma'(\sigma''a)\in\interp{\sigma'(\sigma''\phi)}^+$, since the
definitions of $\interp{\cdot}$ and $\interp{\cdot}^+$ coincide when
the context is empty.  By the induction hypothesis we have
$\sigma'\circ\sigma'' \in\interp{\Gamma'}_{\Delta}$.  So we may
reapply the rule to obtain the desired
$(\sigma'\circ\sigma'')\cup\{(x,\sigma'
a)\}\in\interp{\Gamma',x:\phi}$.

\begin{lemma}[Closability]
\label{lem:close} Suppose the following main assumption is true: for any $\Delta\,\textit{Ok}$
and $\sigma\in\interp{\Gamma}_{\Delta}$, we have $(\sigma
a)\in\interp{\sigma \phi}_{\Delta}$.  In this case, for any such
$\Delta$ and $\sigma$, we also have $(\sigma a)\in\interp{\sigma
  \phi}^+_{\Delta}$.
\end{lemma}

\noindent Assume an arbitrary $\sigma'\in\interp{\Delta}$.  We must
show $\sigma'(\sigma a)\in\interp{\sigma'(\sigma\phi)}$.  By Composing
Substitutions (Lemma~\ref{lem:csubst}), we have
$\sigma'\circ\sigma\in\interp{\Gamma}$.  So we may instantiate the
main assumption with $\sigma'\circ\sigma$ to obtain the the required
formula.

\subsection{The Proof}

The proof of the Soundness Theorem is by induction on the structure of
the assumed typing derivation.  We consider all cases, and implicitly
start each by assuming an arbitrary $\sigma\in\interp{\Gamma}_\Delta$.
We often will use this $\sigma$ to instantiate universal formulas
obtained by application of our induction hypothesis, without
explicitly noting that we are instantiating the induction hypothesis.
If $\sigma$ is a substitution, we will write $\sigma[a'/x]$ for the
substitution that extends $\sigma$ by mapping $x$ to $a'$.

\ 

\noindent \textbf{Case:}

\

$\infer{\Gamma\vdash x:\phi}{\Gamma(x) \equiv \phi}$

\ 

\noindent We prove $\sigma x\in\interp{\sigma \Gamma(x)}^+_{\Delta}$ by
inner induction on the structure of
$\sigma\in\interp{\Gamma}_{\Delta}$.  The base case cannot arise,
since we have $\Gamma(x)$ defined.  For the step case, we have:
\[
\infer{\sigma'\cup\{(y,a)\}\in\interp{\Gamma',y:\phi}_{\Delta}}
      {a\in\interp{\sigma' \phi}^+_{\Delta} & 
      \sigma' \in\interp{\Gamma'}_{\Delta}}
\]
\noindent If $x\equiv y$, then we have $\sigma x\in\interp{\sigma'
  \Gamma(x)}_{\Delta}$ from the first premise.  We just need to show
$\sigma \Gamma(x) \equiv \sigma' \Gamma(x)$.  But this follows from
the fact that $x\not\in\textit{FV}(\phi)$ (by $\Gamma\ \textit{Ok}$).
If $x\not\equiv y$, then by the inner induction hypothesis we have
$\sigma' x\in\interp{\sigma' \Gamma'(x)}_{\Delta}$.  We must show that
this implies the desired $\sigma x\in\interp{\sigma
  \Gamma(x)}_{\Delta}$.  We certainly have $\sigma' x \equiv \sigma
x$, and $\Gamma'(x)\equiv\Gamma(x)$.  So it suffices to show that
$\sigma' \Gamma(x)\equiv \sigma\Gamma(x)$.  But $\Gamma(x)$ cannot
contain $x$, so this holds.

\ 

\noindent \textbf{Case:}

\

$\infer{\Gamma\vdash \join : a = a'}{a \downarrow a'}$

\ 

\noindent If $a\downarrow a'$, we certainly also have $\sigma a
\downarrow \sigma a'$, since joinability is closed under substitution.
This gives us $\Delta\vdash \join : \sigma a = \sigma a'$.  Again by
closure of joinability under substitution, we have $\sigma a
\sim_\Delta \sigma a$, since for any $\sigma'\in\interp{\Delta}$, we
certainly have $\sigma'(\sigma a)\downarrow \sigma'(\sigma a')$.  We
obviously have $\join\in\SN$, so we conclude
$\join\in\interp{\sigma(a_1 = a_2)}_{\Delta}$.

\ 

\noindent \textbf{Case:}

\

$\infer{\Gamma\vdash a:[a''/x]\phi}{\Gamma \vdash a''': a' = a'' & \Gamma \vdash a:[a'/x]\phi & x\not\in\textit{dom}(\Gamma)}$

\ 

\noindent The required conclusion follows by \textbf{R-Join} from
$\sigma a\in\interp{[\sigma a'/x](\sigma\phi)}_\Delta$, which we have
from the induction hypothesis for the second premise.  To enable this
use of \textbf{R-Join}, we need $\Delta\vdash \sigma a''':\sigma a' =
\sigma a''$ and $\sigma a'\sim_\Delta \sigma a''$.  The former we
obtain from $\sigma a'''\in\interp{\sigma a' = \sigma a''}_\Delta$,
which we have by the induction hypothesis for the first premise.  The
latter we obtain as follows.  Consider an arbitrary
$\sigma'\in\interp{\Delta}$.  From this and the fact that
$\sigma\in\interp{\Gamma}_\Delta$, we have
$\sigma'\circ\sigma\in\interp{\Gamma}$ by Composing Substitutions
(Lemma~\ref{lem:csubst}).  

Since $\sigma'\circ\sigma\in\interp{\Gamma}$, we can use it to
instantiate the induction hypothesis for the first premise.  This
gives us $\sigma' (\sigma a''')\in\interp{\sigma'(\sigma a') =
  \sigma'(\sigma a'')}$, which implies $\cdot\vdash\sigma'(\sigma
a'''):\sigma'(\sigma a') = \sigma'(\sigma a'')$.  So consider the
unique normal form $n$ of $\sigma' (\sigma a''')$, which exists by
confluence and \textbf{R-SN}.  We have $n\in\interp{\sigma'(\sigma a')
  = \sigma'(\sigma a'')}$ by repeated application of \textbf{R-Pres}.
This implies $n:\sigma'(\sigma a') = \sigma'(\sigma a'')$.  By
Progress, this $n$ must be a value.  We may now apply Canonical Forms
(Lemma~\ref{lem:canonical}), to conclude that $n \equiv \join$.  Now
by the definition of $\interp{\sigma'(\sigma a') = \sigma'(\sigma
  a'')}$, we have $\sigma'(\sigma a') \downarrow \sigma'(\sigma a'')$,
as required.  We assumed an arbitrary $\sigma'\in\interp{\Delta}$, so
we may conclude that $\sigma'(\sigma a') \downarrow \sigma'(\sigma
a'')$ holds for all such $\sigma'$.  This is sufficient for $\sigma a'
\sim_\Delta \sigma a''$.

\ 

\noindent \textbf{Case:}

\

$\infer{\Gamma \vdash a : \forall x:\phi'.\phi}
       {\Gamma, x:\phi' \vdash a:\phi & x\not\in\textit{FV}(a)}$

\ 

\noindent From the
induction hypothesis, we infer the following, for any
$\sigma'\in\interp{\Gamma,x:\phi'}_{\Delta'}$, for any
$\Delta'\subset\sigma (\Gamma,x:\phi')$:
\begin{equation}
\label{eqlama}
\forall a'\in\interp{\sigma \phi'}^+_{\Delta}.\ (\sigma[a'/x]) a\in\interp{(\sigma[a'/x])\phi}_{\Delta}
\end{equation}

\noindent Our first need is to prove $a\in\SN$.  For this, we
instantiate (\ref{eqlama}) with $\Delta'\equiv \Delta,x:\sigma\phi'$;
$\sigma'\equiv\sigma[x/x]$; and $a'\equiv x$.  Note that it is at
precisely this point that we critically need open substitutions in the
statement of Soundness.  To show that this instantiation is legal, we
must, of course, prove that $\sigma'\in\interp{\Gamma}_{\Delta'}$.
For this, we need two things.  First, we need to know that
$x\in\interp{\sigma \phi'}^+_{\Delta,x:\sigma\phi'}$.  This follows
because $x\in\interp{\sigma \phi'}_{\Delta,x:\sigma\phi'}$ by
\textbf{R-Prog} (since $\next(a) = \emptyset$); and further, for any
$\sigma'\in\interp{\Delta,x:\sigma\phi'}$, we derive from that same
fact the formula $\sigma' x\in\interp{\sigma'(\sigma\phi')}$, which we
require for closability of $x$.  Now we use the following lemma (proof
in Section~\ref{sec:weaksubst} below) to finish our proof of the
intermediate fact $\sigma'\in\interp{\Gamma,x:\phi'}_{\Delta'}$:

\begin{lemma}[Weakening Substitutions]
\label{lem:weaksubst}
If $\sigma\in\interp{\Gamma}_{\Delta}$ and
$\Delta,y:\phi'\,\textit{Ok}$, then
$\sigma\in\interp{\Gamma}_{\Delta,y:\phi'}$.
\end{lemma}

\noindent Using this intermediate fact
$\sigma'\in\interp{\Gamma,x:\phi'}_{\Delta'}$, we may indeed
instantiate (\ref{eqlam}) above with $\sigma'$, $\Delta'$ and $x$ for
$a'$, as mentioned.  This gives us $\sigma a\in\interp{\sigma
  \phi}_{\Delta'}$. By \textbf{R-SN}, we then obtain $\sigma
a\in\interp{\sigma\phi}_{\Delta,x:\sigma\phi'}$.  From this, we obtain
$\sigma a\in\SN$ and $\Delta\vdash \lambda x.a : \Pi x:\sigma
\phi'.\sigma \phi$, which we need to show $\sigma a\in\interp{\Pi
  x:\sigma\phi'.\sigma\phi}_\Delta$.

To complete this case, it suffices to consider arbitrary
$a'\in\interp{\sigma\phi'}^+_{\Delta}$, and show $\sigma
a\in\interp{[a'/x]\sigma\phi}_{\Delta}$.  Instantiating (\ref{eqlama})
with $a'$, we obtain $(\sigma[a'/x])
a\in\interp{(\sigma[a'/x])\phi}_{\Delta}$.  This is equivalent to the
goal, thanks to the following facts about the substitutions in
question.  Since $x\not\in\textit{FV}(a)$, we have $(\sigma[a'/x])a
\equiv \sigma a$.  Also, we have $x\not\in\textit{ran}(\sigma)$ by the
following lemma.  So we get the desired $\sigma
a\in\interp{[a'/x]\sigma\phi}_{\Delta}$.

\begin{lemma}[Basic Property of Substitutions]
\label{lem:substvars}
$\sigma\in\interp{\Gamma}_{\Delta}\ \land \Gamma \mathit{Ok}\ \To\ \sigma(x)\in \interp{\sigma \Gamma(x)}^+_{\Delta}$
\end{lemma}
\noindent The proof is by induction on the structure of the assumed
derivation.

\ 

\noindent \textbf{Case:}

\

$\infer{\Gamma \vdash a : [a'/x]\phi}{\Gamma \vdash a:\forall x:\phi'.\phi & \Gamma \vdash a':\phi'}$

\ 

\noindent This follows immediately from induction hypothesis, and the
definition of $\interp{\cdot}$ for $\forall$-types (here, the type
$\forall x:\sigma \phi'.\sigma \phi$), using the fact that various
substitutions involved commute, as in cases above.  We critically use
Closability (Lemma~\ref{lem:close}), to get $\sigma
a'\in\interp{\sigma\phi'}^+_\Delta$ from $\sigma
a'\in\interp{\sigma\phi'}_\Delta$.

\ 

\noindent \textbf{Case:}

\

$\infer{\Gamma \vdash \lambda x.a : \Pi x:\phi'.\phi}
      {\Gamma, x:\phi' \vdash a:\phi}$

\ 

\noindent We begin just as for the \texttt{(spec-abs)} case.  From the
induction hypothesis, we infer the following, for any
$\sigma'\in\interp{\Gamma,x:\phi'}_{\Delta'}$, for any
$\Delta'\subset\sigma (\Gamma,x:\phi')$:
\begin{equation}
\label{eqlam}
\forall a'\in\interp{\sigma' \phi'}^+_{\Delta'}.\ (\sigma'[a'/x]) a\in\interp{(\sigma'[a'/x])\phi}_{\Delta'}
\end{equation}
\noindent By the same reasoning as for the \texttt{(spec-abs)} case,
we obtain $\sigma a\in\SN$ and $\Delta\vdash \lambda x.a : \Pi
x:\sigma \phi'.\sigma \phi$.

Now by the definition of $\interp{\cdot}$, it suffices to prove that
for all $a'\in\interp{\sigma\phi'}^+_{\Delta}$, we have $((\lambda
x.(\sigma a))\ a') \in \interp{[a'/x](\sigma \phi)}_{\Delta}$.  We
prove that (\ref{eqlam}) implies this, by inner induction on
$\nu(\sigma a)+\nu(a')$, which is defined by \textbf{R-SN} (for
$\sigma a$ and $a'$).  By \textbf{R-Prog}, it suffices to prove
$\next((\lambda x.(\sigma a))\ a') \subset
\interp{[a'/x](\sigma\phi)}_{\Delta}$, since the term in question
(i.e., $((\lambda x.(\sigma a))\ a')$) is neutral and appropriately
typable since $\sigma a$ is.  In more detail, since we have $\sigma
a\in\interp{\sigma \phi}_{\Delta,x:\sigma\phi'}$, we obtain
$\Delta,x:\sigma\phi'\vdash \sigma a:\sigma \phi$.  Then we apply the
typing rule for $\lambda$-abstractions to obtain $\Delta\vdash \lambda
x.\sigma a:\Pi x:\sigma\phi'.\sigma \phi$, and we conclude the typing
proof with the application rule on this fact and $\Delta\vdash
a':\sigma\phi'$.

Now the possibilities for reduction of the term in question are
summarized by:
\[
\next((\lambda x. \sigma a)\ a') \ \subset\ 
((\lambda x. \next(\sigma a))\ a')\ \cup \ 
((\lambda x. \sigma a)\ \next(a'))\ \cup \ 
\{ [a'/x]\sigma a \}
\]

\noindent We have $((\lambda x.\next(\sigma
a))\ a')\subset\interp{[a'/x](\sigma\phi)}_\Delta$ by the inner
induction hypothesis, using \textbf{R-Pres} to conclude $\next(\sigma
a)\subset\interp{\sigma\phi'}_{\Delta}$.  For the set $((\lambda
x. \sigma a)\ \next(a'))$, we use the inner induction hypothesis to
conclude that for all $a''\in\next(a')$, we have $((\lambda x. \sigma
a)\ a'')\in\interp{[a''/x](\sigma \phi)}_{\Delta}$.  Applying the
induction hypothesis here requires the fact that
$a''\in\interp{\sigma\phi'}_\Delta$, which follows by \textbf{R-Pres}.
Now we may apply \textbf{R-Join} with $a' \sim_\Delta a''$ and
$\Delta\vdash\join:a' = a''$ (which follow from $a'\leadsto a''$), to
obtain $((\lambda x. \sigma a)\ a'')\in\interp{[a'/x](\sigma
  \phi)}_{\Delta}$.  This implies $((\lambda x. \sigma
a)\ \next(a'))\subset\interp{[a'/x](\sigma \phi)}_{\Delta}$, as
required.

Finally, we wish to conclude $[a'/x]\sigma a\in\interp{[a'/x](\sigma
  \phi')}_{\Delta}$ by instantiating (\ref{eqlam}) above with
$\Delta'\equiv \Delta$; $\sigma'\equiv\sigma$; and $a'\equiv a'$. We
have the required $a'\in\interp{\sigma\phi'}^+_\Delta$, of course.
But we also need the fact that $(\sigma[a'/x]) \phi' = [a'/x](\sigma
\phi')$.  This holds because $x\not\in\textit{ran}(\sigma)$ (by
Lemma~\ref{lem:substvars}, as in an earlier case).  So we conclude the
desired $[a'/x]\sigma a\in\interp{[a'/x](\sigma \phi')}_{\Delta}$.

\ 

\noindent \textbf{Case:}

\

$\infer{\Gamma \vdash (a\ a') : [a'/x]\phi}{\Gamma \vdash a:\Pi x:\phi'.\phi & \Gamma\vdash a':\phi'}$

\ 

\noindent By the induction hypothesis, we have $\sigma a\in\interp{\Pi
  x:\sigma\phi'.\sigma\phi}_\Delta$ and $\sigma
a'\in\interp{\sigma\phi'}_\Delta$.  By Closability
(Lemma~\ref{lem:close}), we then get $\sigma
a'\in\interp{\sigma\phi'}^+_\Delta$.  Then using the definition of
$\interp{\cdot}$ at $\Pi$-type, we directly obtain $((\sigma
a)\ (\sigma a'))\in\interp{[\sigma a'/x]\sigma\phi}_\Gamma$.  Since
$x\not\in\textit{ran}(\sigma)$ (by Lemma~\ref{lem:substvars}), we get
from this the desired conclusion, namely $\sigma
(a\ a')\in\interp{\sigma([a'/x]\phi)}_\Gamma$.

\ 

\noindent \textbf{Case:}

\

$\infer{\Gamma\vdash 0:\nat}{\ }$

\

\noindent Since $0$ is a normal form, we have $0\in\SN$ and
$\Delta\vdash 0:\nat$, which suffices for this case.

\ 

\noindent \textbf{Case:}

\

$\infer{\Gamma\vdash \nil:\langle\vc\ \phi\ 0\rangle}{\ }$

\ 

\noindent Since $\nil$ is a normal form, we have $\nil\in\SN$, and of
course, $\Delta\vdash\nil:\langle\vc\ \sigma \phi\ 0\rangle$. For the
second conjunct of the definition of $\interp{\cdot}$ at $\vc$-type,
we have $0\downarrow 0$. For the third conjunct, assume
$\nil\leadsto^*(\cons\ a\ a')$ for some $a$ and $a'$.  This is easily
shown to be impossible, since reduction cannot possibly turn $\nil$
into a $\cons$-term.

\ 

\noindent \textbf{Case:}

\

$\infer{\Gamma\vdash (S\ a):\nat}{\Gamma\vdash a:\nat  }$

\ 

\noindent By the induction hypothesis, we have $\sigma
a\in\interp{\nat}_\Delta$, which is equivalent to the conjunction of
$\sigma a\in\SN$ and $\Delta\vdash \sigma a:\nat$.  This implies
$(S\ \sigma a)\in\SN$ and $\Delta\vdash (S\ \sigma a):\nat$, which
suffices.

\ 

\noindent \textbf{Case:}

\

$\infer{\Gamma\vdash (R_\nat\ a\ a'\ a''):[a''/x]\phi}
      {\begin{array}{l}\Gamma \vdash a'' : \nat \\
       \Gamma \vdash a : [0/x]\phi \\
       \Gamma \vdash a' : \Pi y:\nat. \Pi u : [y/x]\phi. [(S y)/x]\phi
       \end{array}}$

\ 

\noindent By the induction hypothesis, we have
\begin{itemize}
\item $\sigma a''\in\interp{\nat}_\Delta$
\item $\sigma a\in\interp{\sigma[0/x]\phi}_\Delta$
\item $\sigma a'\in\interp{\Pi y:\nat.\Pi u:\sigma([y/x]\phi).\ \sigma([(S y)/x]\phi)}_\Delta$
\end{itemize}

\noindent We will prove that for any $b\in\interp{\nat}_\Delta$, and
assuming the second two of these facts, we have $(R_\nat\ (\sigma
a)\ (\sigma a')\ b) \in\interp{[b/x]\sigma\phi}_\Delta$. The proof is
by inner induction on the measure $\nu(\sigma a)+\nu(\sigma
a')+\nu(b)+l(b)$.  Our measure is defined, since all the terms
involved are reducible and hence strongly normalizing by
\textbf{R-SN}.  By \textbf{R-Prog}, it suffices to prove
$\next(R_\nat\ (\sigma a)\ (\sigma a')\ b)
\subset\interp{[b/x]\sigma\phi}_\Delta$, since the term in question is
neutral and appropriately typable.  The possibilities for reduction are
summarized by:
\[
\begin{array}{lll}
(R_\nat\ (\sigma a)\ (\sigma a')\ b) & \subset & (R_\nat\ \next(\sigma a)\ (\sigma a')\ b)\ \cup\\
\ &\ & (R_\nat\ (\sigma a)\ \next(\sigma a')\ b)\ \cup \\
\ &\ & (R_\nat\ (\sigma a)\ (\sigma a')\ \next b)\ \cup \\
\ &\ & \{ (\sigma a)\ |\ b \equiv 0 \}\ \cup \\
\ &\ & \{ ((\sigma a')\ b'\ (R_\nat\ (\sigma a)\ (\sigma a')\ b'))\ |\ b \equiv (S\ b')\}
\end{array}
\]

\noindent The first three cases are for when the reduction is due to
reduction in a subterm.  The second two are for when the term in
question is itself a redex.  For the first two cases, we use the inner
induction hypothesis and \textbf{R-Pres}.  For the third, we do the
same, except also apply \textbf{R-Join} with $b\sim_\Delta b'$ for
$b'\in\next(b)$.  This ensures that we have $(R_\nat\ (\sigma
a)\ (\sigma a')\ \next(b))\subset\interp{[b/x]\sigma\phi}_\Delta$ (the
critical point being that we have $b$ in the type, and not some
$b'\in\next(b)$).  The fourth case follows by our assumption that
$\sigma a\in\interp{[0/x]\phi}_\Delta$ (note that in this case that
the type in question is equivalent to the desired $[b/x]\phi$).  For
the fifth case, we have $(R_\nat\ (\sigma a)\ (\sigma
a')\ b')\in\interp{[b'/x]\sigma\phi}_\Gamma$ by the inner induction
hypothesis, using the fact that $b\in\SN$ and $b = (S\ b')$ implies
$b'\in\SN$; and this then implies $b'\in\interp{\nat}_\Delta$ by
definition of $\interp{\cdot}$.  Note that we obtain $\Delta\vdash
b':\nat$ from $\Delta \vdash (S\ b'):\nat$, by applying Simplifying
Inversion (Lemma~\ref{lem:simplinv} above).  By the definition of
$\interp{\cdot}$ at $\Pi$-type and our hypothesis that $\sigma a'$ is
reducible at the appropriate $\Pi$-type, we have that the given term
is in the set $\interp{[(S\ b')/x]\phi}_\Delta$, which is equal to the
desired $\interp{[b/x]\phi}_\Gamma$.

\ 

\noindent \textbf{Case:}

\

$\infer{\Gamma\vdash (\cons\ a\ a'):\langle \vc\ \phi\ (S\ l)\rangle}
      {\begin{array}{l}\Gamma\vdash a:\phi \\ \Gamma \vdash a':\langle \vc\ \phi\ l\rangle
       \end{array}}$

\ 

\noindent By the induction hypothesis, we have $\sigma
a\in\interp{\phi}_\Delta$ and $\sigma a'\in\interp{\langle\vc\ \sigma
  \phi\ \sigma l\rangle}_\Gamma$.  By \textbf{R-SN}, these facts imply
$\sigma a\in\SN$ and $\sigma a'\in\SN$, respectively, and hence
$\sigma(\cons\ a\ a')\in\SN$.  We also have $\Delta\vdash
\sigma(\cons\ a\ a'):\langle\vc\ \sigma \phi\ (S\ \sigma l)\rangle$.
We must show the conjuncts of the definition of $\interp{\cdot}$ at
$\vc$-type to conclude $(\cons\ (\sigma a)\ (\sigma
a'))\in\interp{\langle \vc\ \phi\ (S\ l)\rangle}_\Delta$.  The second
conjunct is vacuously true, since we cannot have
$(\cons\ a\ a')\leadsto^*\nil$. The third conjunct follows directly
from our assumptions.

\ 

\noindent \textbf{Case:}

\

$\infer{\Gamma\vdash (R_\vc\ a\ a'\ a''):[l/y, a''/x]\phi}
      {\begin{array}{l}\Gamma \vdash a'' : \langle \vc\ \phi'\ l\rangle \\
       \Gamma \vdash a : [0/y,\nil/x]\phi \\
       \Gamma \vdash a' : \Pi z:\phi'. \forall l:\nat. \Pi v :\langle \vc\ \phi'\ l\rangle. \Pi u : [l/y, v/x]\phi. \\
        \ \ \ \ \ \ \ \ \ \ \ \ \  [(S\ l)/y, (\cons\ z\ v)/x]\phi
       \end{array}}$

\ 

\noindent This case is similar to that for $R_\nat$ above, although it
is for this case that we have the various clauses of the definition of
$\interp{\cdot}$ at $\vc$-type.  By the induction hypothesis, we have
\begin{itemize}
\item $\sigma a''\in\interp{\sigma\langle \vc\ \phi'\ l\rangle}_\Delta$
\item $\sigma a\in\interp{\sigma[0/y,\nil/x]\phi}_\Delta$
\item $\sigma a'\in\interp{\sigma\Pi z:\phi'. \forall l:\nat. \Pi v :\langle \vc\ \phi'\ l\rangle. \Pi u : [l/y, v/x]\phi.[(S\ l)/y, (\cons\ z\ v)/x]\phi}_\Delta$
\end{itemize}

\noindent It is sufficient to prove that for any $l$, for any
$b\in\interp{\langle\vc\ \phi'\ l\rangle}_\Delta$, and assuming the
second two of these facts, we have $(R_\vc\ (\sigma a)\ (\sigma
a')\ b) \in\interp{[l/y,b/x]\sigma\phi}_\Delta$. The proof is by inner
induction on the measure $\nu(\sigma a)+\nu(\sigma a')+\nu(b)+l(b)$.
As above, this measure is defined, by \textbf{R-SN}.  By
\textbf{R-Prog}, it suffices to prove $\next(R_\vc\ (\sigma
a)\ (\sigma a')\ b) \subset\interp{[l/y,b/x]\sigma\phi}_\Delta$, since
the term in question is neutral and appropriately typable.  The
possibilities for reduction are summarized by:
\[
\begin{array}{lll}
(R_\vc\ (\sigma a)\ (\sigma a')\ b) & \subset & (R_\vc\ \next(\sigma a)\ (\sigma a')\ b)\ \cup\\
\ &\ & (R_\vc\ (\sigma a)\ \next(\sigma a')\ b)\ \cup \\
\ &\ & (R_\vc\ (\sigma a)\ (\sigma a')\ \next b)\ \cup \\
\ &\ & \{ (\sigma a)\ |\ b \equiv \nil \}\ \cup \\
\ &\ & \{ ((\sigma a')\ b'\ b''\ (R_\vc\ (\sigma a)\ (\sigma a')\ b''))\ |\ b \equiv (\cons\ b'\ b'')\}
\end{array}
\]

\noindent The first three cases are for when the reduction is due to
reduction in a subterm.  The second two are for when the term in
question is itself a redex.  For the first two cases, we use the inner
induction hypothesis and \textbf{R-Pres}.  For the third, we also
apply \textbf{R-Join} as in the $R_\nat$ case above, to ensure that we
have $(R_\vc\ (\sigma a)\ (\sigma
a')\ \next(b))\subset\interp{[b/x]\sigma\phi}_\Gamma$.  The fourth
case follows by our assumption that $\sigma
a\in\interp{[0/y,\nil/x]\phi}_\Gamma$.  By the definition of
$\interp{\cdot}$ at $\vc$-type, we have $l\sim_\Gamma 0$; so we can
apply \textbf{R-Join} and the fact that $b=\nil$ to obtain
$a\in\interp{[l/y,b/x]\phi}_\Gamma$, as required.

We now consider the fifth case. First, since $b = (\cons\ b'\ b'')$
and $b\in\interp{\langle\vc\ \phi'\ l\ \rangle}_\Gamma$, the
definition of $\interp{\cdot}$ at $\vc$-type gives us the following
facts for some $l'$:
\begin{itemize}
\item $b'\in\interp{\phi'}_\Delta$
\item $b''\in\interp{\langle\vc\ \phi'\ l'\rangle}_\Delta$
\item $l\sim_\Gamma (S\ l')$
\end{itemize}
\noindent We next apply the inner induction hypothesis to obtain
$(R_\vc\ (\sigma a)\ (\sigma
a')\ b'')\in\interp{[l'/y,b''/x]\sigma\phi}_\Delta$; this is legal,
since the measure has decreased (in particular, $l(b'')<l(b)$).  With
this obtained, we use the definition of $\interp{\cdot}$ at $\Pi$-type
and our hypothesis that $\sigma a'$ is reducible at the appropriate
$\Pi$-type.  So we obtain the fact that the given term is in the set
$\interp{[(S\ l'))/y,(\cons\ b'\ b'')/x]\phi}_\Delta$.  We can then
use \textbf{R-Join} with the fourth assumed formula above, and the
fact that $b = (\cons\ b'\ b'')$, to get that the term is in the
desired $\interp{[l/y,b/x]\phi}_\Delta$.

\subsection{Proof of Weakening Substitutions (Lemma~\ref{lem:weaksubst})}
\label{sec:weaksubst}

\noindent The proof is by induction on the structure of the assumed
derivation of $\sigma\in\interp{\Gamma}_{\Delta}$.  The base case is
trivial.  For the step case, we have:
\[
\infer{\sigma'\cup\{(x,a)\}\in\interp{\Gamma',x:\phi}_{\Delta}}
      {a\in\interp{\sigma'\phi}^+_{\Delta} & \sigma' \in\interp{\Gamma'}_{\Delta}}
\]
\noindent The induction hypothesis gives us
$\sigma'\in\interp{\Gamma'}_{\Delta,y:\phi'}$.  We just need
$a\in\interp{\sigma'\phi}^+_{\Delta,y:\phi'}$, and we can
reapply the rule to get the desired result.  For this, we use the
following lemma:

\begin{lemma}[Weakening for Closable Terms]
\label{lem:weakclose}
Suppose $\Delta,y:\phi'\,\textit{Ok}$.  Then
$a\in\interp{\phi}^+_{\Delta}$ implies
$a\in\interp{\phi}^+_{\Delta,y:\phi'}$.
\end{lemma}

\noindent Now we apply the following lemma to complete the proof,
noting that the variable $y$ of interest here is not in the free
variables of $a$ or $\phi$, and so the assumption implies the required
universal formula:

\begin{lemma}[Weakening-Strengthening for Interpretations]
\label{lem:weakstr}
Suppose $a\in\interp{\phi}^+_\Delta$.  Then we have
$a\in\interp{\phi}^+_{\Delta,y:\phi'}$ iff for all
$a'\in\interp{\phi'}^+_\Delta$, we have $[a'/y]a\in\interp{[a'/y]
  \phi}^+_{\Delta}$.
\end{lemma}

\noindent The proof makes frequent use of the following lemma, which
we prove briefly first:

\begin{lemma}[Weakening-Strengthening for Ground Joinability]
\label{lem:weakstrj}
Suppose $(\Delta,y:\phi')\,\textit{Ok}$.  Then we have
$a_1\sim_{\Delta,y:\phi'} a_2$ iff for all
$a'\in\interp{\phi'}_\Delta$, we have $[a'/y] a_1\sim_{\Delta}
[a'/y]a_2$.
\end{lemma}

\noindent First, suppose $a_1 \sim_{\Delta,y:\phi'} a_2$, and consider
arbitrary $\sigma'\in\interp{\Delta}$ and
$a'\in\interp{\sigma'\phi'}$.  Then we have
$(\sigma'\circ[a'/y])\in\interp{\Delta,y:\phi'}$ by Composing
Substitutions (Lemma~\ref{lem:csubst}).  We can then use $a_1
\sim_{\Delta,y:\phi'} a_2$ to get $(\sigma'([a'/y]a_1))\downarrow
(\sigma'([a'/y]a_2$, as required.  

Second, suppose that for all $a'\in\interp{\phi'}_\Delta$, we have
$[a'/y]a_1\sim_\Delta [a'/y]a_2$, and show $a_1\sim_{\Delta,y:\phi'}
a_2$.  Assume arbitrary $\sigma\in\interp{\Delta,y:\phi'}$.  Then for
some $\sigma'$ and $a'\in\interp{\sigma'\phi}$, we have
$\sigma\equiv\sigma'[a'/y]$.  Instantiating our assumption with this
$a'$ and then $\sigma'$, we obtain the desired conclusion.

We now turn to the main proof for Weakening-Strengthening, which is by
induction on $(|\Gamma|,d(\phi),l(a))$.  In all cases, the typing
statement in question follows by either Weakening
(Lemma~\ref{lem:weak}) or Substitution (Lemma~\ref{lem:subst}), so we
omit consideration of typing below.  Strong normalization of the
substitution instances follows from the definition of closability (and
the assumption that $a$ is a closable term in context $\Delta,y:\phi'$).


\

\noindent \textbf{Case:} $\phi\equiv\nat$.

\

\noindent This case is trivial.

\ 

\noindent \textbf{Case:} $\phi\equiv\langle\vc\ \phi\ l\rangle$.

\

\noindent These cases follow easily by Weakening-Strengthening for Ground
Joinability (Lemma~\ref{lem:weakstrj}) and the induction hypothesis.

\ 

\noindent \textbf{Case:} $\phi\equiv\Pi x:\psi.\psi'$.

\

\noindent First, assume $a\in\interp{\phi}^+_{\Delta,y:\phi'}$, and
show $[a'/y]a\in\interp{[a'/y] \phi}^+_{\Delta}$ for an arbitrary
$a'\in\interp{\phi'}^+_\Delta$.  It suffices to consider arbitrary
$a''\in\interp{[a'/y]\psi}^+_\Delta$, and show
$(([a'/y]a)\ a'')\in\interp{[a''/x][a'/y]\psi'}^+_\Delta$.  Since
$y\not\in\textit{FV}(a'')$, we certainly have
$[\hat{a}/y]a''\in\interp{[\hat{a}/y][a'/y]\psi}^+_\Delta$ for all
$\hat{a}\in\interp{\phi'}^+_\Delta$.  So we may apply the induction
hypothesis to conclude $a''\in\interp{[a'/y]\psi}^+_{\Delta,y:\phi'}$.
Now we may use our assumption of reducibility of $a$ in context
$\Delta,y:\phi'$ to conclude
$(a\ a'')\in\interp{[a''/y]\psi'}_{\Delta,y:\phi'}$.  To obtain
$(a\ a'')\in\interp{[a''/y]\psi'}^+_{\Delta,y:\phi'}$ from this,
assume an arbitrary partition
$(\Delta_1,\Delta_2)\equiv(\Delta,y:\phi')$, and arbitrary
$\sigma\in\interp{\Delta_2}_{\Delta_1}$.  We must show $\sigma
(a\ a'')\in\interp{\sigma [a''/y]\psi'}_{\Delta_1}$.  If $\Delta_2$ is
empty, this statement is equivalent to the fact
$(a\ a'')\in\interp{[a''/y]\psi'}_{\Delta,y:\phi'}$, which we already
have.  So suppose $\Delta_2$ ends in $y:\phi'$.  Then
$y\not\in\textit{ran}(\sigma)$.  Also, $y\not\in\textit{FV}(a'')$, so
our current goal formula is equivalent to $((\sigma
a)\ a'')\in\interp{[a''/y](\sigma\psi')}_{\Delta_1}$.  Instantiating
our assumption of closability of $a$ with
$\sigma|_{\textit{dom}(\Delta_1)}$, we obtain $\sigma
a\in\interp{\sigma\phi}_{\Delta_1}$.  This is then sufficient for the
desired conclusion, since we easily obtain $\sigma
a''\in\interp{\sigma [a'/y]\psi}^+_{\Delta_1}$ from our assumption of
closability of $a''$ in context $\Delta$.  Having obtained $(a\ a'')$
closable, we may now apply the induction hypothesis again, to obtain
$(([a'/y]a)\ a'')\in\interp{[a'/y][a''/x]\psi'}_\Delta$, noting again
that $y\not\in\textit{FV}(a'')$.  Closability of $a$ and $a''$ again
imply closability of this final term.

Now assume that for all $a'\in\interp{\phi}^+_\Delta$, we have
$[a'/y]a\in\interp{[a'/y] \phi}^+_{\Delta}$; and show
$a\in\interp{\phi}^+_{\Delta,y:\phi'}$.  It suffices to consider
arbitrary $a''\in\interp{\psi}^+_{\Delta,y:\phi'}$, and show
$(a\ a'')\in\interp{[a''/x]\psi'}_{\Delta,y:\phi'}$.  By the induction
hypothesis, we have $[a'/y]a''\in\interp{[a'/y]\psi}^+_\Delta$ for any
$a'\in\interp{\phi'}_\Delta$.  Consider arbitrary such $a'$.  We have
$([a'/y]a\ [a'/y]a'')\in\interp{[[a'/y]a''/x][a'/y]\psi'}_\Delta$ by
reducibility of $[a'/y]a$ in context $\Delta$.  Closability of $a$ and
$a''$ in context $(\Delta,y:\phi')$ again imply closability of this
term. This is true for any $a'$, so we may apply the induction
hypothesis again to conclude
$(a\ a'')\in\interp{[a''/x]\psi'}_\Delta$, and again obtain
closability as above, for the required conclusion.

\

\noindent \textbf{Case:} $\phi\equiv\forall x:\psi.\psi'$.

\

\noindent This case is very similar to the previous one, so we omit it.

\ 

\noindent \textbf{Case:} $a_1 = a_2$.

\

\noindent This follows by Weakening-Strengthening for Ground
Joinability (Lemma~\ref{lem:weakstrj}).

\section{Proof of Corollaries of Theorem~\ref{thm:soundness}}

\subsection{Proof of Strong Normalization (Corollary~\ref{cor:sn})}

By Soundness for Interpretations,
we have $\sigma a\in\interp{\sigma\phi}_\Delta$ for all $\Delta$ and
$\sigma$ with $\Delta\subset\sigma\Gamma$ and
$\sigma\in\interp{\Gamma}_\Delta$.  We instantiate this by taking
$\Gamma$ for $\Delta$ and the identity substitution $\textit{id}$ on
$\textit{dom}(\Gamma)$ for $\sigma$.  We have
$\textit{id}\in\interp{\Gamma}_\Gamma$, since for all
$x\in\textit{dom}(\Gamma)$, we have $x\in\interp{\Gamma(x)}^+_\Gamma$
by \textbf{R-Prog} and the fact that if $\sigma'\in\interp{\Gamma}$,
then by that assumption, we get $\sigma' x\in\interp{\sigma'
  \Gamma(x)}$, which is needed for closability of $x$.  This
instantiation yields $a\in\interp{\phi}_\Gamma$, which implies
$a\in\SN$ by \textbf{R-SN}.

\subsection{Proof of Equational Soundness (Theorem~\ref{cor:eqsnd})}

By Type Preservation, Progress, and
Canonical Forms, we obtain $a\leadsto^*\join$.  By Inversion, the only
possible derivations are \texttt{(conv)} inferences starting with a
$\join$-introduction.  This implies $b_1\downarrow b_2$, because
joinability is closed under substitution of joinable terms.  An easy
corollary is:

\section{Proofs for section \ref{sec:tveclarge} (Large Eliminations)}

\subsection{Proof of Critical Properties}

\

\noindent \textbf{R-Canon}. If $a\in\interp{\phi}$, then $a\leadstov^* v$ for some $v$. Furthermore,
if $\phi$ is a value type (i.e. $\nat$, $\Pi$, $\forall$, $=$, or $\vc$), then $v$ is the corresponding introduction form.

\begin{proof} Immediate from the definition of $\interp{\ }$.
\end{proof}

\noindent \textbf{R-Pres}. If $a\in\interp{\phi}$ and $a \leadstov a'$, then $a' \in \interp{\phi}$.

\begin{proof}Induction on the depth of $\phi$.

The clauses of the form $a \leadstov^*$ are all proven in the same way: for instance if $a \leadstov^* n$ and $a \leadstov a'$, then $a' \leadstov^* n$ by determinacy of $\leadstov$. This takes care of all cases except $\Pi$ and $\forall$.

For $\Pi y:\phi'.\phi$, we also need to show $\forall a''\in\interp{\phi'}.\ (a'\ a'')\in\interp{[a''/x]\phi}$. Let $a''\in\interp{\phi'}$. By assumption we know $(a a'') \in \interp{[a''/x]\phi}$. But $(a a'') \leadstov (a' a'')$, so by the IH at the type $[a''/x]\phi$ (which is of lower depth) $(a' a'') \in \interp{[a''/x]\phi}$ as required.

The case $\forall y:\phi'.\phi$ is similar to the above case: we need to show $(a'\ \impapp)\in\interp{[a''/x]\phi}$ and use that $(a\ \impapp) \leadstov (a'\ \impapp)$.
\end{proof}

\noindent \textbf{R-Prog}. If $a \leadstov a'$, and $a' \in \interp{\phi}$, then $a\in\interp{\phi}$.

\begin{proof}
Induction on the depth of $\phi$.

The clauses of the form $a \leadstov^* v$ are all proven in the same way: for instance if $a' \leadstov^* n$ and $a \leadstov a'$, then $a' \leadstov^* n$. This takes care of all cases except $\Pi$ and $\forall$.

For $\Pi y:\phi'.\phi$, we also need to show $\forall a''\in\interp{\phi'}.\ (a\ a'')\in\interp{[a''/x]\phi}$. Let $a''\in\interp{\phi'}$. By assumption $(a'\ a'') \in \interp{[a''/x]\phi}$. But $(a\ a'') \leadstov (a'\ a'')$, so by IH at the type $[a''/x]\phi$ (which is of lower depth), $(a\ a'') \in \interp{[a''/x]\phi}$ as required.

The case $\forall y:\phi'.\phi$ is similar to the above case: we need to show $(a\ \impapp) \in \interp{[a''/x]\phi}$ and use $(a\ \impapp) \leadstov (a'\ \impapp)$.

\end{proof}

\noindent \textbf{R-Join}. If $a_1 \downarrow a_2$, then $a \in \interp{[a_1/x]\phi}$ 
implies $a \in \interp{[a_2/x]\phi}$.

\begin{proof}
Induction on the  depth of $\phi$. 
\begin{itemize}
\item $\nat$. Trivially true since $[a_1/x]\nat = [a_2/x]\nat = \nat$.
\item $\langle\vc\ \phi \l\rangle$.  By assumption $a \in
  \interp{\langle\vc\ \phi \l}$, so either $a \leadstov^* \nil$ and
  $[a_1/x]l \leadsto^* 0$, or $a \leadstov^* (\cons\ v\ v')$ and
  $[a_1/x]l \leadsto^* (S\ n)$ with $v \in \interp{[a_1/x]\phi}$ and
  $v' \in \interp{\langle\ [a_1/x]\phi\ n\rangle}$.

In the first case, note that joinability implies $[a_2/x]l \leadsto^* 0$.
In the second case, joinability gives $[a_2/x] \leadsto^* (S\ n)$,
and the IH gives $v \in \interp{[a_2/x]\phi}$ and $v' \in \interp{\langle\vc\ 
[a_2/x]\phi\ n\rangle}$.

\item $\Pi y:\phi'.\phi$. The first conjunct is the same for both $[a_1/x]\phi$ and $[a_2/x]\phi$.
For the second conjunct, let $a' \in \interp{[a_2/x]\phi'}$. By IH $a' \in \interp{[a_1/x]\phi'}$,
so $(a\ a') \in \interp{[a'/y][a_1/x]\phi}$. Since $y$ was a bound variable we can choose it such
that $a' \not\in \operatorname{FV}(a_1)\cup\{x\}$,
so $[a'/y][a_1/x]\phi = [a_1/x][a'/y]\phi$. By IH applied to $[a'/y]\phi$ (which is of lower depth),
$(a a') \in \interp{[a_2/x][a'/y]\phi}$ as required.

\item $\forall y:\phi'.\phi$. Similar to the previous case.

\item $b_1 = b_2$. We need to show that $[a_1/x]b_1 \downarrow [a_1/x]b_2$ implies $[a_2/x]b_1 \downarrow [a_2/x]b_2$, which is true.

\item $(\ifzero\ b\ \phi_1\ (\alpha.\phi_2))$. Note that $[a_1/x]b \leadsto^* n$ implies $[a_2/x]b \leadsto^* n$, and then by the IH.
\end{itemize}
\end{proof}

\subsection{Proof of Theorem \ref{thm:fundamental_tveclarge} (Fundamental Lemma for Large Eliminations version of $\interp{\ }$)}

\ 

\noindent \textbf{Case:}

\

$\infer{\Gamma\vdash x:\phi}{\Gamma(x) \equiv \phi}$

\ 

\noindent Immediate by $\sigma \in \interp{\Gamma}$.

\ 

\noindent \textbf{Case:}

\

$\infer{\Gamma\vdash \join : a = a'}
{
  \Gamma\vdash a : \phi
& \Gamma\vdash a' : \phi'
& a \downarrow a'}$

\ 

\noindent $\join$ is a value of the right form. We get $\sigma a
\downarrow \sigma a'$, since joinability is closed under substitution.
We get $\exists v_i.\sigma a_i \leadstov^* v_i$ by IH and \textbf{R-Canon}.

\ 

\noindent \textbf{Case:}

\

$\infer{\Gamma\vdash a:[a''/x]\phi}
{\Gamma \vdash a''': a' = a'' & \Gamma \vdash a:[a'/x]\phi & x\not\in\textit{dom}(\Gamma)}$

\ 

\noindent By the IH from the second premise we have $\sigma
a\in\interp{[\sigma a'/x](\sigma\phi)}$.  By the IH from the first
premise we have $\sigma a''' \in \interp{\sigma(a'=a'')}$, so $\sigma
a' \downarrow \sigma a''$.  So by \textbf{R-Join}, $\sigma a \in \interp{[\sigma a''/x]/\psi} = \interp{[\sigma a''/x](\sigma\phi)}$.

\ 

\noindent \textbf{Case:}

\

$\infer{\Gamma \vdash (\lambda a) : \forall x:\phi'.\phi}
       {\Gamma, x:\phi' \vdash a:\phi & x\not\in\textit{FV}(a)}$

\ 

\noindent $(\lambda a)$ is a value of the right form. We must show
$(\lambda \sigma a) \in \interp{\forall x:\sigma\phi'.\sigma\phi}$.

Consider some $a' \in \interp{\sigma\phi'}$. By \textbf{R-Prog}, it
suffices to show $a \in \interp{[a'/x]\sigma\phi}$, since $((\lambda
a)\ \impapp)\leadstov a$.  Let $\sigma' = \sigma \cup \{(x,a')\}$. Then
$\sigma' \in \interp{\Gamma, x:\phi'}$, so by IH we have $\sigma' a
\in \interp{\sigma' \phi}$, that is $\sigma a \in
\interp{[a'/x]\sigma'\phi}$

\ 

\noindent \textbf{Case:}

\

$\infer{\Gamma \vdash a\ \impapp : [a'/x]\phi}
 {\Gamma \vdash a:\forall x:\phi'.\phi & \Gamma \vdash a':\phi'}$

\ 

\noindent This follows immediately from induction hypothesis, and the
definition of $\interp{\cdot}$ for $\forall$-types.

\ 

\noindent \textbf{Case:}

\

$\infer{\Gamma \vdash \lambda x.a : \Pi x:\phi'.\phi}
      {\Gamma, x:\phi' \vdash a:\phi}$

\ 

\noindent $\lambda x.a$ is a value of the right form.  We must show
$(\lambda x. a) \in \interp{\Pi x:\sigma\phi'.\sigma\phi}$.

Consider some $a' \in \interp{\sigma \phi'}$. We must show $(\lambda x.\sigma a) a' \in
\interp{[a'/x]\sigma \phi}$. By \textbf{R-Prog} it suffices to show that it steps to a term in $\interp{[a'/x]\sigma \phi}$.

By \textbf{R-Canon}, $a' \leadsto v'$ for some $v'$. We proceed by the number
of steps $a'$ takes to normalize. In the base case $a'$ is already a
value. Then $(\lambda x. \sigma a) a' \leadstov [a'/x]\sigma a = \sigma' a$ where $\sigma' = \sigma \cup \{(x,a')\}$.
$\sigma' \in \interp{\Gamma, x:\phi'}$, so by IH $\sigma' a \in \interp{\sigma'\phi}$. 

In the step case, $a' \leadstov a''$ for some $a''$, so $(\lambda x.\sigma a) a' \leadstov (\lambda x.\sigma a) a''$. By \textbf{R-Pres}, $a'' \in \interp{\sigma \phi'}$, so the inner IH applies and $(\lambda x.\sigma a) a'' \in \interp{[a''/x]\sigma\phi'}$. But $a' \leadstov a''$, so $a \leadsto a''$, so $a' \downarrow a''$, so \textbf{R-Join} applies and $(\lambda x.\sigma a) a'' \in \interp{[a'/x]\sigma\phi'}$ as required.

\ 

\noindent \textbf{Case:}

\

$\infer{\Gamma \vdash (a\ a') : [a'/x]\phi}{\Gamma \vdash a:\Pi x:\phi'.\phi & \Gamma\vdash a':\phi'}$

\ 

\noindent This follows immediately from induction hypothesis, and the
definition of $\interp{\cdot}$ for $\Pi$-types.

\ 

\noindent \textbf{Case:}

\

$\infer{\Gamma\vdash 0:\nat}{\ }$

\

\noindent$0$ is a value of the right form.

\ 

\noindent \textbf{Case:}

\

$\infer{\Gamma\vdash (S\ a):\nat}{\Gamma\vdash a:\nat  }$

\ 

\noindent By the induction hypothesis, we have $\sigma
a\in\interp{\nat}$, so by \textbf{R-Canon}  $\sigma a\leadstov^*
n$. Then $(S\ \sigma a)\leadstov^* (S n)$, which is a value of the
right form.

\ 

\noindent \textbf{Case:}

\

$\infer{\Gamma\vdash (R_\nat\ a\ a'\ a''):[a''/x]\phi}
      {\begin{array}{l}\Gamma \vdash a'' : \nat \\
       \Gamma \vdash a : [0/x]\phi \\
       \Gamma \vdash a' : \Pi y:\nat. \Pi u : [y/x]\phi. [(S y)/x]\phi
       \end{array}}$

\ 

\noindent By the induction hypothesis, we have
\begin{itemize}
\item $\sigma a''\in\interp{\nat}$
\item $\sigma a\in\interp{\sigma[0/x]\phi}$
\item $\sigma a'\in\interp{\Pi y:\nat.\Pi u:\sigma([y/x]\phi).\ \sigma([(S y)/x]\phi)}$
\end{itemize}

\noindent We will prove that for any $b\in\interp{\nat}$, and
assuming the second two of these facts, we have $(R_\nat\ (\sigma
a)\ (\sigma a')\ b) \in\interp{[b/x]\sigma\phi}$. The proof is
by inner induction on the measure $\nu(\sigma a)+\nu(\sigma
a')+\nu(b)+l(b)$.  Our measure is defined, since all the terms
involved are normalizing by \textbf{R-Canon}.

By \textbf{R-Prog}, it suffices to prove that $R_\nat\ (\sigma a)\
(\sigma a')\ b$ steps to a term in $\interp{[b/x]\sigma\phi}$.  The
terms $(\sigma a)$, $(\sigma a')$, and $b$ are all in  $\interp{\cdot}$, 
so by \textbf{R-Canon} each of them either steps or is a value. By
considering the cases, one of the following must be the case:
\[
\begin{array}{llll}
R_\nat\ (\sigma a)\ (\sigma a')\ b & \leadstov & R_\nat\ c\ (\sigma a')\ b &\qquad\text{ where }(\sigma a) \leadstov c \\
R_\nat\ v\ (\sigma a')\ b & \leadstov & R_\nat\ v\ c'\ b &\qquad\text{ where }(\sigma a') \leadstov c' \\
R_\nat\ v\ v'\ b & \leadstov & R_\nat\ v\ v'\ c &\qquad\text{ where } b \leadstov c \\
R_\nat\ v\ v'\ 0 & \leadstov & v \\
R_\nat\ v\ v'\ (S\ n) & \leadstov & v'\ n\ (R_\nat\ v\  v'\  n)\\
\end{array}
\]

\noindent The first three cases are for when the reduction is due to
reduction in a subterm.  The second two are for when the term in
question is itself a redex.  For the first two cases, we use the inner
induction hypothesis and \textbf{R-Pres}.  For the third, we do the
same, except also apply \textbf{R-Join} with $b\downarrow c$.  This
ensures that we have $(R_\nat\ (\sigma a)\ (\sigma a')\ b)\in
\interp{[b/x]\sigma\phi}$ (the critical point being that we have $b$
in the type, and not $c$) The fourth case follows by our assumption
that $\sigma a\in\interp{[0/x]\phi}$ (note that in this case
that the type in question is equivalent to the desired $[b/x]\phi$).
For the fifth case, we have $(R_\nat\ (\sigma a)\ (\sigma a')\
n)\in\interp{[n/x]\sigma\phi}$ by the inner induction
hypothesis. Since $n$ is a number we trivially have $n \in
\interp{\nat}$.  By the definition of $\interp{\cdot}$ at $\Pi$-type
and our hypothesis that $\sigma a'$ is reducible at the appropriate
$\Pi$-type, we have that the given term is in the set $\interp{[(S\
  n)/x]\phi}$, which is equal to the desired
$\interp{[b/x]\phi}_\Gamma$.

\ 

\noindent \textbf{Case:}

\

$\infer{\Gamma\vdash \nil:\langle\vc\ \phi\ 0\rangle}{\ }$

\ 

\noindent $\nil$ and $0$ are values of the right form.

\ 

\noindent \textbf{Case:}

\

$\infer{\Gamma\vdash (\cons\ a\ a'):\langle \vc\ \phi\ (S\ l)\rangle}
      {\begin{array}{l}\Gamma\vdash a:\phi \\ \Gamma \vdash a':\langle \vc\ \phi\ l\rangle
       \end{array}}$

\ 

\noindent We prove the second disjunct of $\sigma (\cons\ a\ a') \in \interp{\sigma\langle \vc\ \phi\ (S\ l)\rangle}$. 
By IH and \textbf{R-Canon}, we know  $\sigma a$ and $\sigma a'$
reduce to some values $v$ and $v'$. Then $\sigma(\cons\ a\ a') \leadstov^* (\cons\ v\ v')$ as required. Similarly
from the IH we know $\sigma l \leadstov^* n$, so $\sigma(S\ l) \leadstov^* (S\ n)$ as required.

\

\noindent \textbf{Case:}

\

$\infer{\Gamma\vdash (R_\vc\ a\ a'\ a''):[l/y, a''/x]\phi}
      {\begin{array}{l}\Gamma \vdash a'' : \langle \vc\ \phi'\ l\rangle \\
       \Gamma \vdash a : [0/y,\nil/x]\phi \\
       \Gamma \vdash a' : \Pi z:\phi'. \forall l:\nat. \Pi v :\langle \vc\ \phi'\ l\rangle. \Pi u : [l/y, v/x]\phi. \\
        \ \ \ \ \ \ \ \ \ \ \ \ \  [(S\ l)/y, (\cons\ z\ v)/x]\phi
       \end{array}}$

\ 

\noindent This case is similar to that for $R_\nat$ above. By the induction hypothesis, we have
\begin{itemize}
\item $\sigma a''\in\interp{\sigma\langle \vc\ \phi'\ l\rangle}$
\item $\sigma a\in\interp{\sigma[0/y,\nil/x]\phi}$
\item $\sigma a'\in\interp{\sigma\Pi z:\phi'. \forall l:\nat. \Pi v :\langle \vc\ \phi'\ l\rangle. \Pi u : [l/y, v/x]\phi.[(S\ l)/y, (\cons\ z\ v)/x]\phi}$
\end{itemize}

\noindent It is sufficient to prove that for any $l$, for any
$b\in\interp{\langle\vc\ \phi'\ l\rangle}$, and assuming the
second two of these facts, we have $(R_\vc\ (\sigma a)\ (\sigma
a')\ b) \in\interp{[l/y,b/x]\sigma\phi}$. The proof is by inner
induction on the measure $\nu(\sigma a)+\nu(\sigma a')+\nu(b)+l(b)$.
As above, this measure is defined, by \textbf{R-Canon}. 

By \textbf{R-Prog}, it suffices to prove that $R_\vc\ (\sigma a)\
(\sigma a')\ b$ steps to a term in $\interp{[l/y,b/x]\sigma\phi}$.  The
terms $(\sigma a)$, $(\sigma a')$, and $b$ are all in  $\interp{\cdot}$, 
so by \textbf{R-Canon} each of them either steps or is a value. By
considering the cases, one of the following must be the case:
\[
\begin{array}{llll}
R_\vc\ (\sigma a)\ (\sigma a')\ b & \leadstov & R_\vc\ c\ (\sigma a')\ b &\qquad\text{ where }(\sigma a) \leadstov c \\
R_\vc\ v\ (\sigma a')\ b & \leadstov & R_\vc\ v\ c'\ b &\qquad\text{ where }(\sigma a') \leadstov c' \\
R_\vc\ v\ v'\ b & \leadstov & R_\vc\ v\ v'\ c &\qquad\text{ where } b \leadstov c \\
R_\vc\ v\ v'\ \nil & \leadstov & v \\
R_\vc\ v\ v'\ (\cons\ u\ u') & \leadstov & v'\ u\ u'\ (R_\vc\ v\  v'\  u')\\
\end{array}
\]

\noindent The first three cases are for when the reduction is due to
reduction in a subterm.  The second two are for when the term in
question is itself a redex.  For the first two cases, we use the inner
induction hypothesis and \textbf{R-Pres}.  For the third, we also
apply \textbf{R-Join} as in the $R_\nat$ case above, to ensure that we
have $(R_\vc\ (\sigma a)\ (\sigma a')\ c)\in\interp{[l/y,b/x]\sigma\phi}$.  The fourth
case follows by our assumption that $\sigma
a\in\interp{[0/y,\nil/x]\phi}$.  By the definition of
$\interp{\cdot}$ at $\vc$-type, we must have $l\leadstov^* 0$; so we can
apply \textbf{R-Join} and the fact that $b=\nil$ to obtain
$a\in\interp{[l/y,b/x]\phi}$, as required.

For the fifth case, we know by assumption that $(\cons\ u\ u') \in
\interp{\sigma\langle\vc\ \phi\ l\rangle}$.  By the definition of
$\interp{\cdot}$ that means that $u \in \interp{\phi}$, $\sigma l
\leadstov^* (S\ n)$, and $u' \in \interp{\langle\vc\ \phi\
n\rangle}$.

Then we have $(R_\vc\ (\sigma a)\ (\sigma a')\
u')\in\interp{[n/y,u'/x]\sigma\phi}$ by the inner induction
hypothesis. By the definition of $\interp{\cdot}$ at $\Pi$-type and
our hypothesis that $\sigma a'$ is reducible at the appropriate
$\Pi$-type, we have that the given term is in the set $\interp{[(S\
  n)/y, (\cons\ u\ u')/x]\sigma\phi}$. By using \textbf{R-Join} on
$l\downarrow(S\ n)$, this implies the desired $\interp{[l/y,(\cons\ u\
  u')/x]\phi}$.

\ 

\noindent \textbf{Case:}

\

$\infer[\texttt{foldS}]
{\Gamma\vdash a : \ifzero\ (S a')\ \phi\ \phi'}
{\Gamma\vdash a : \phi' & \Gamma\vdash a':nat}$

\ 

\noindent By IH we have $\sigma a' \in \interp{\nat}$, so by \textbf{R-Canon}, $\sigma a' \leadstov^* n$ for some $n$. Therefore 
$\sigma (S a') \leadsto^* (S\ n)$, so we need to show $\sigma a \in \interp{\sigma\phi'}$, which we get by IH.

\

\noindent \textbf{Case:}

\

$\infer[\texttt{unfoldS}]
{\Gamma\vdash a : \phi'}
{\Gamma\vdash a : \ifzero\ (S\ a')\ \phi\ \phi' & \Gamma\vdash a':\nat}$

\ 

\noindent Similar to the previous case.

\noindent \textbf{Case:}

\

$\infer
{\Gamma\vdash a : \ifzero\ 0\ \phi\ \phi'}
{\Gamma\vdash a : \phi}$

\ 

\noindent Similar to \texttt{unfoldS} case.

\

\noindent \textbf{Case:}

\

$\infer
{\Gamma\vdash a : \phi}
{\Gamma\vdash a : \ifzero\ 0\ \phi\ \phi'}$

\ 

\noindent Similar to \texttt{foldS} case.

\

\fi 

\end{document}